\documentclass[manuscript]{acmart}

\usepackage{hyperref}       
\usepackage{url}            
\usepackage{booktabs}       
\usepackage{nicefrac}       
\usepackage{microtype}      
\usepackage{xcolor}         
\usepackage{amsmath, amsfonts, dsfont, amsthm}
\usepackage{blindtext}
\usepackage{mathtools}
\usepackage{comment}
\usepackage{algorithm, algpseudocode}
\usepackage{caption, subcaption}
\usepackage[export]{adjustbox}
\usepackage{multirow}
\usepackage{graphicx}
\usepackage{booktabs}
\usepackage{varwidth}
\usepackage{multicol}
\graphicspath{ {./images/} }
\usepackage{xr}
\usepackage{babel}

\newcommand{\p}{\mathbb{P}}
\newcommand{\Reals}{\mathbb{R}}

\newcommand{\se}{{\sf se}}
\newcommand{\majority}{\text{majority}}
\newcommand{\minority}{\text{minority}}
\newcommand{\innocent}{\text{inn}}
\newcommand{\criminal}{\text{cri}}
\newcommand{\majo}{\text{maj}}
\newcommand{\mino}{\text{min}}

\newcommand{\systemic}{systematic~}

\newcommand{\doop}[1]{\text{do(}#1\text{)}}

\newtheorem{definition}{Definition}
\newtheorem{theorem*}{Theorem}
\newtheorem{proposition}{Proposition}

\newtheorem{lemma}{Lemma}

\AtBeginDocument{%
  \providecommand\BibTeX{{%
    \normalfont B\kern-0.5em{\scshape i\kern-0.25em b}\kern-0.8em\TeX}}}

\setcopyright{acmcopyright}
\copyrightyear{2023}
\acmYear{2023}
\acmDOI{XXXXXXX.XXXXXXX}

\begin{document}

\title{A Causal Framework to Evaluate Racial Bias in Law Enforcement Systems}


\author{Jessy Xinyi Han}
\email{xyhan@mit.edu}
\authornote{Corresponding author.}
\affiliation{%
  \institution{Massachusetts Institute of Technology}
  \country{USA}
 }
 
\author{Andrew Miller}
\affiliation{%
  \institution{U.S. Naval Academy}
  \country{USA}
 }
 
 \author{S. Craig Watkins}
\affiliation{%
  \institution{University of Texas at Austin}
  \country{USA}
 }

\author{Christopher Winship}
\affiliation{%
  \institution{Harvard University}
  \country{USA}
}

 \author{Fotini Christia}
\affiliation{%
  \institution{Massachusetts Institute of Technology}
  \country{USA}
 }

\author{Devavrat Shah}
\affiliation{%
  \institution{Massachusetts Institute of Technology}
  \country{USA}
 }


\begin{abstract}

We are interested in developing a data-driven method to evaluate race-induced biases in law enforcement systems. While the recent works have addressed this question in the context of police-civilian interactions using {\em police stop} data, they have two key limitations.
First, bias can only be properly quantified if true {\em criminality} is accounted for in addition to race, but it is absent in prior works\footnote{In our opinion, this is the primary reason for the disagreement between \cite{KLM} and \cite{Gaebler_Cai_Basse_Shroff_Goel_Hill_2020} in terms of \textcolor{black}{the identification strategies of} the causal estimand corresponding to racial bias and subsequently a not so pleasant public debate between them.}. 
Second, law enforcement systems are multi-stage and hence it is important to isolate the true source of bias within the ``causal chain of interactions'' rather than simply focusing on the end outcome; this can help guide reforms. 

In this work, we address these challenges by presenting a multi-stage causal framework incorporating criminality. We provide a theoretical characterization and an associated data-driven method to evaluate (a) the presence of any form of racial bias, and 
(b) if so, the primary source of such a bias in terms of race and criminality. 
Our framework identifies three canonical scenarios with distinct characteristics: in settings like 
(1) airport security, the primary source of observed bias against a race is likely to be bias in law enforcement against {\em innocents} of that race; 
(2) AI-empowered policing\footnote{\textcolor{black}{We do not intend to propose the use of AI tools for policing here. Any potential use of AI tools for policing should be heavily monitored and regulated.}}, the primary source of observed bias against a race is likely to be bias in law enforcement against {\em criminals} of that race; 
and (3) police-civilian interaction, the primary source of observed bias against a race could be bias in law enforcement against that race or bias from the general public in reporting (e.g. via 911 calls) against {\em the other} race. 
Through an extensive empirical study using police-civilian interaction (stop) data and 911 call data, we find an instance of such a counter-intuitive phenomenon: in New Orleans, the observed bias is {\em against} the majority race and the likely reason for it is the over-reporting (via 911 calls) of incidents involving the minority race by the general public. 
%

\end{abstract}





\maketitle

\section{Introduction}

The role played by race, and more generally, gender, religion, and ethnicity, in decision-making within societal systems has garnered renewed interests over the past decade. Among various societal systems, law enforcement plays a crucial role in maintaining public safety and order but is, like any other, susceptible to systemic bias \cite{skogan_fairness_2004}. Racial disparities in law enforcement are believed to perpetuate cycles of poverty, exclusion, and marginalization of minority communities. Therefore, a thorough examination of racial disparities and guidance for reforms to reduce, or ideally eliminate, such disparities is required \cite{goff_racial_2012, hellman_measuring_2020, white_misdemeanor_2019}. 

Towards this, the research community has been focusing on benchmarking data, understanding the underlying causes of biases and proposing policy changes to mitigate them \cite{NEURIPS_DATASETS_AND_BENCHMARKS2021_7f6ffaa6, goncalves_few_2021, blair_community_2021}. 
In particular, racial disparities in law enforcement are often associated with a higher level of law enforcement actions against the discriminated race \cite{gelman2007analysis, ridgeway2007analysis}. Indeed, such an outcome specific inspection led \citep{Fryer2019} to conclude a lack of racial bias against minorities on lethal use of force. Subsequent works like \citep{KLM,Gaebler_Cai_Basse_Shroff_Goel_Hill_2020} have pointed out the need for a causal framework to correct for  bias in such analyses and proposed appropriate causal estimands to evaluate the presence of racial bias or disparity. However, these prior works have two key limitations. 

First, they do not explicitly account for the notion of {\em criminality}, i.e. whether a crime was committed. To evaluate racial disparity, it is important to identify whether the use of force was justified, which requires the context whether a crime was committed. Indeed, in our opinion, the absence of criminality in the prior works was the primary reason for the public debate to arise around the \textcolor{black}{identification} of causal estimand to evaluate racial bias or disparity (see \cite{debateblog}). 

Second, law enforcement is a multi-stage system. While the use of force by police is the eventual measurement to determine racial disparity, the incidents in which the police interact with civilians are determined by a sequence of decisions. This sequence includes which incidents are reported through 911 calls, how police are dispatched, and whether police stop civilians involved and eventually exert force against them. While \cite{KLM,Gaebler_Cai_Basse_Shroff_Goel_Hill_2020} have introduced one of the stages of this multi-stage system compared to \citep{Fryer2019}, further stages need to be introduced.\\

\noindent \textbf{Contributions.} In this work, we address these two challenges by presenting a multi-stage causal framework incorporating criminality. While criminality often remains unobserved or unverifiable, its inclusion in our framework enables us to define the appropriate causal estimand associated with racial disparity or bias. The proof is, however, in the pudding: the need of using {\em observed} data to evaluate the presence of racial disparity. 

Towards that, we posit the following natural hypothesis (or assumption): given the {\em context} which may represent social, economic and all other conditions that might lead an individual to commit a crime, the race of an individual and the act of committing a crime are (conditionally) independent. In other words, the reasons behind someone committing a crime are determined by the surrounding context, not by their race. 

Under this hypothesis, despite the notion of criminality being unobserved, we are able to define a test statistics of {\em observational racial disparity} (see \eqref{eq:delta.u} for a precise definition) using observed data only. The {\em observational racial disparity} evaluates whether there is racial disparity but falls short in identifying the primary cause of the disparity - which race is discriminated against, whether the criminals or innocents, and who is primarily responsible for such a discrimination. Indeed, in general such identification is beyond reach. However, we identify three canonical scenarios under which the primary cause of racial disparity can be discerned: 
 
\noindent {\em (1) Scenario 1} corresponds to settings like airport security checks, where everyone entering the system undergoes a check (i.e. interaction with law enforcement). In this scenario, if there is observational racial disparity based on the test statistic, it is primarily due to bias in law enforcement actions against innocent individuals of that race. See Section \ref{sec:mostly_innocent} for details. 

\noindent {\em (2) Scenario 2} corresponds to settings like AI-empowered policing where an individual is suspected and put through further interaction with law enforcement only if there is a high chance of requiring further investigation. In this scenario, if there is observational racial disparity based on the test statistic, then it is primarily due to the bias in law enforcement actions against criminal individuals of that race. See Section \ref{sec:mostly_guilty} for details. 

\noindent {\em (3) Scenario 3} corresponds to settings like police-civilian interactions where biases can arise from public reporting through 911 calls and/or law enforcement by police. In this scenario, we find that observed racial bias against one race may come from policing actions against that race or the general public's over-reporting of the other race. See Section \ref{sec:more_or_less} for details. 

We apply our framework to the law enforcement system with police-civilian interaction using the police stop data and 911 call data for New York City (NYC) and New Orleans over a number of years. Through careful processing and stitching such data, we construct the multi-stage dataset of various incidences, corresponding to Scenario 3 shown above. Upon careful evaluation of the {\em observational racial disparity} for both NYC and New Orleans, we find the following: 

In NYC, there are observational biases against minority races (Black/Hispanic) compared to the majority (White), which could be explained by the racial disparity in policing actions. However, in New Orleans we find the exact opposite - an observational bias against the majority race as per the observed test statistic. This could be explained by the over-reporting of incidents involving minorities through 911 calls, as suggested by our framework under Scenario 3. This leads us to a counter-intuitive phenomenon - even though at its core there is bias against the minority (in reporting), the observational data suggests a bias against the majority! \\

\noindent{\bf Related Prior Works.}
A growing literature with causal inference has been developed to model multi-stage processes in sociology \cite{elwert_endogenous_2014}, healthcare \cite{huang_analysis_2012}, and political science \cite{slough_phantom_2022}. Among them, an active area of research is the investigation of racial bias in law enforcement processes. \textcolor{black}{Various methods have been been proposed \citep{doi:10.1146/annurev-criminol-011518-024731}, including the benchmark test \citep{gelman2007analysis}, outcome test \citep{Fryer2019}, threshold test \citep{simoiu2017problem} and so on. 
} 

Most previous works evaluating potential bias in law enforcement systems focus on individual stages of the process, including traffic stops \cite{pierson_large-scale_2020,antonovics_new_2009,grogger_testing_2006}, vehicle searches \cite{knowles_racial_2001,ridgeway_assessing_2006}, verbal communication \cite{voigt_language_2017}, arrests \cite{mitchell_examining_2015}, sentencing \cite{alesina_test_2014}, and use of force \cite{nix_birds_2017,kramer_stop_2018}, among other outcomes. Some recent studies do take a systematic approach to evaluate racial disparities. Specifically, \cite{KLM,Gaebler_Cai_Basse_Shroff_Goel_Hill_2020} utilize a causal framework to refine the simplistic approach in \citep{Fryer2019} to evaluate racial disparity based on the use of force by police in police-civilian interactions with the police stop data. However, these works have two limitations as mentioned earlier: they do not account for the notion of {\em criminality} and fall short of encompassing some stages of law enforcement, including how incidents are reported. 

When comparing our causal framework and the observational racial disparity statistics with \cite{KLM,Gaebler_Cai_Basse_Shroff_Goel_Hill_2020}, several key differences emerge. First, the causal frameworks in \cite{KLM,Gaebler_Cai_Basse_Shroff_Goel_Hill_2020} are subsets of our causal framework as they do not account for criminality and other stages of causal framework. Second, when the framework is restricted to the simpler setting in \cite{Gaebler_Cai_Basse_Shroff_Goel_Hill_2020, KLM}, our test statistics induces similar test statistics suggested in their works. However, \cite{Gaebler_Cai_Basse_Shroff_Goel_Hill_2020} requires an additional ignorability assumption to justify the identification of their statistics while our work naturally derives such justification from the hypothesis that race and criminality of an individual are conditionally independent given the context. Like \cite{Gaebler_Cai_Basse_Shroff_Goel_Hill_2020}, \cite{KLM} needs to make additional assumptions to justify their data-driven approach for point-identification of their causal estimand, unlike our work. Our framework identifies three canonical scenarios where observational racial disparity can be explained through the primary source of racial disparity across stages of law enforcement and the types of disparity (against race and criminals / innocents).


Since differential crime reporting rates have significant impacts on policing systems \citep{DBLP:journals/corr/abs-2102-00128}, another salient feature of our work is the ability to leverage the multi-stage nature of law enforcement by including public reporting through 911 call data. Some existing studies have leveraged the reporting party to assess possible discrimination but they focus on 911 calls themselves and are primarily concerned with teasing out biases based on officer-level characteristics. For example, \cite{Hoekstra_Sloan2022} uses 911 calls as a natural experiment to assess the effect of officer race on use of force likelihood. \cite{weisburst_whose_2022} leverages 911 calls to evaluate how ``officer arrest propensity'' influences police discretion. These works do not leverage---nor aim to do so---the reporting decisions to measure levels of overall racial disparities in law enforcement actions.

~~\\
\noindent{\bf Organization.} The remainder of the article is organized as follows:
In Section \ref{sec:systemic}, we lay out a \systemic causal framework for evaluating bias in law enforcement. 
In Section \ref{ssec:priorworks}, we explain how prior works fits within this framework.
In Section \ref{sec:racial_disparity}, we discuss conditions under which we could identify the primary source of racial disparities in different law enforcement systems. 
In Section \ref{sec:empirics}, we explain our empirical strategy, analyze our data and present the results. 
In Section \ref{sec:discussion}, we conclude with a discussion on the potential extensions of our framework and reflect on the braoder implications of our work.

\section{A Causal Framework for Racial (Dis)Parity}\label{sec:systemic}

Law enforcement systems are complex and often involve multiple stages. To comprehend the various intricacies of these systems, including the presence of racial disparities, we need to closely study the inter-relationship between these stages. In this section, we introduce a multi-stage causal framework. 

\subsection{Causal Framework} 

 \begin{minipage}{\textwidth}
  \begin{minipage}[b]{0.49\textwidth}
    \centering
    \includegraphics[height=0.1\textheight, valign=t]{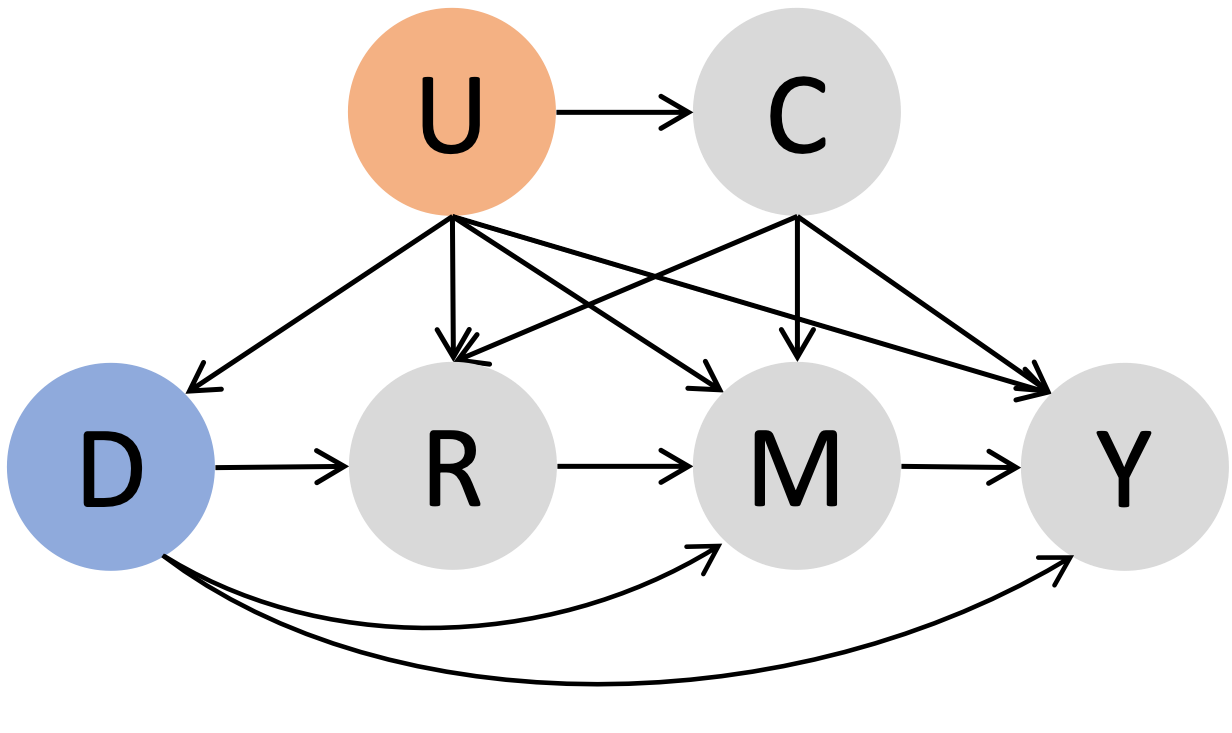}
    \captionof{figure}{A causal DAG correpsonding to multi-stage law enforcement system.} 
    \label{fig:systematic_DAG}
  \end{minipage}
  \hfill
  \begin{minipage}[b]{0.49\textwidth}
    \centering
    \begin{tabular}{|c|c|c|}
    \hline
        Notation & Meaning & Observed?\\ \hline
        $D$ & race & yes\\ \hline
        $R$ & report & yes\\ \hline
        $M$ & stop & yes\\ \hline
        $Y$ & policing outcome & yes\\ \hline
        $U$ & context & partially yes\\ \hline
        $C$ & criminality & no\\ \hline        
    \end{tabular}
    \captionof{table}{A summary of notations.}
    \label{tab:notations}
    \end{minipage}
  \end{minipage}

Consider a collection of incidents that demand the attention of law enforcement agencies. As mentioned earlier, law enforcement is a {\em multi-stage} process, involving not only the law enforcement agency but also the reporting party who may report only a portion of the incidents. Consequently, only the reported incidents will receive appropriate law enforcement attention and outcomes. For example, in airport security checks, individuals intending to board a plane self-report to Transportation Security Administration (TSA) officers and undergo mandatory security checks; in advanced AI-empowered policing, AI algorithms detect traces of criminal activities and report them to the police who later enforce policing actions; in a typical police-civilian encounter, as we find using empirical analysis (see Section \ref{sec:empirics} for details), that the majority of incidents that come to law enforcement's attention originate from 911 calls. Once the interaction is initiated in any of these scenarios, the law enforcement agents decide whether to stop and further question the suspects and, depending on the interaction, whether to proceed with potentially severe actions. We formalize this multi-stage law-enforcement process through a natural causal framework represented by a causal Directed Acyclic Graph (DAG) in Figure \ref{fig:systematic_DAG} with pertinent notations summarized in Table \ref{tab:notations}. 

\noindent
{\bf Incident.} An incident involves a person or people of interest, and we assume each incident involves a single person of interest. Let $D \in \{\majority, \minority\}$ denote the race (or type) of the individual associated with the incident. We specifically analyze the effect of race in this work.{\color{black}\footnote{Much scholarship on race has demonstrated the constructivist nature of racial categorizations \citep[]{sen_race_2016}. We are prevented from taking a constructivist approach given data limitations in large-scale law enforcement datasets. Moreover, causally-oriented literature on biases in law enforcement also generally takes a non-constructivist approach, so we follow suit in order to speak to that literature.}}

\noindent
{\bf Reporting.} An individual or incident is reported to the law enforcement system. We use $R = 1$ to indicate an incident is reported and $R = 0$ otherwise.

\noindent
{\bf Stop and Law Enforcement.} Following the reporting of an incident to the law enforcement agency, the individual involved may be stopped for further inquiry. Let $M = 1$ denote the incident that leads to a stop and $M = 0$ otherwise. If the stop occurs ($M = 1$), we denote by $Y = 1$ the enforcement of some subsequent action (e.g., verbal warning, use of force, citation, arrest) and $Y = 0$ otherwise. 

\noindent
{\bf Criminality.} Law enforcement systems aim to hold individuals accountable for crimes committed. Let $C \in \{0,1\}$ denote whether the individual commits the crime ($C = 1$) or not ($C = 0$). 

\noindent{\bf Context.} The context of each incident, such as the socio-economic characteristics and location, that may influence the involved civilian's race, reporting, stopping, and law enforcement actions. 
We denote by $U \in \Reals^d$ the context that impacts the race $D$, criminality $C$, reporting $R$, stop $M$, and law-enforcement action $Y$. 

\noindent{\bf Causal Mechanism.} As noted, $U$ may influence $D, C, R, M$, and $Y$. The true nature of criminality $C$ obviously could affect the reporting $R$, stop $M$, and law-enforcement action $Y$. Reporting $R$ influences the stop $M$, and through $M$, it influences the law-enforcement action $Y$. The race $D$ may affect the reporting $R$, stop $M$, and law-enforcement action $Y$. The underlying model is represented in a causal Directed Acyclic Graph (DAG) in Figure \ref{fig:systematic_DAG}. The key (and only) assumption encoded in this causal mechanism is that there is {\em no edge} between $C, D$. That is, given 
{\em context} $U$, criminality $C$ and race $D$ are {\em conditionally independent}. In words, this encodes 

\begin{center}
    {\em it is the context that makes someone criminal (or not), not their race.}
\end{center}

\subsection{Racial (Dis)Parity}\label{ssec:parity}

Given the causal DAG, we study the ``ideal'' racial parity to understand how the {\em intervention} of race $D$ affects the law-enforcement outcome $Y$\footnote{The notion of manipulating a person's race is understandably controversial. A possible solution is to shift from immutable traits to perceptions of them \citep[]{Greiner_Rubin_2011}. Thus, in this work, we model a scenario where two individuals in the same context are identical except for the perception of their races and understand whether this leads to different law-enforcement actions.}. Specifically, adopting the standard causal notation \citep[][]{Pearl_2014}, we define treatment parity among criminals and innocents.

\begin{definition}[Racial parity among criminals]\label{def:1}
Law enforcement has {\em racial parity among criminals} if 
\begin{align}
\p^{\doop{D=\majority}}[ Y=1 \,\mid\, C=1, U=u] & = \p^{\doop{D=\minority}}[Y=1 \,\mid\, C=1, U=u], ~~\forall u \in \Reals^d \label{eq:parity.crime.do}.
\end{align}

\end{definition}

\begin{definition}[Racial parity among innocents] \label{def:2}
Law enforcement has {\em racial parity among innocents} if 
\begin{align}
\p^{\doop{D=\majority}}[ Y=1 \,\mid\, C=0, U=u] & = \p^{\doop{D=\minority}}[Y=1 \,\mid\, C=0, U=u] , ~~\forall u \in \Reals^d \label{eq:parity.innocent.do}.
\end{align}
\end{definition}
\noindent Here, $\p^{\doop{D = d}}$ represents the distribution induced under the intervention of race $D = d$ for $d \in \{\majority, \minority\}$.  In this work, we aim to understand if either parity holds and, if not, whether there is bias against citizens from the $\majority$ or $\minority$ racial group with respect to criminals or innocents or both. {\color{black} Note that depending on the scenarios, there are many different definitions of fair treatment in literature \citep{10.1145/3194770.3194776, 10.5555/3294771.3294834, 10.5555/3294996.3295162}. The notion of parity introduced in Definition \ref{def:1} and \ref{def:2} is a counterfactual version of \emph{equalized odds} \citep{NIPS2016_9d268236}  \textcolor{black}{often used in criminology literature \citep{530d463c-e065-33e2-abfa-6f61af336c55}}.} Admittedly, analyzing these is not straightforward since many aspects are not fully observable, like criminality $C$ and context $U$. To achieve our goal, first, we discuss some invariances obeyed by such a \systemic causal framework. 

\subsection{Achieving Racial Parity Across Stages} \label{sec:compare_fairness}

Given the multi-stage causal framework and the ideal goal of achieving parity regarding race, we discuss certain structural invariances implied by our setup. 
As per the causal DAG in Fig.  \ref{fig:systematic_DAG}, we could represent the causal quantity $\p^{\doop{D=d}}[ Y=1 \,\mid\, C=c, U=u]$ in Eq. \ref{eq:parity.crime.do} and \ref{eq:parity.innocent.do} as the following and the proof is in Appendix Section \ref{appendix:prop_do_calculus}:

\begin{lemma}\label{lemma:calculation}
    $\forall u \in \Reals^d$, $c \in \{0,1\}$, and $d \in \{\minority, \majority\}$, 
\begin{align}
\p^{\doop{D=d}}[ Y=1 \,\mid\, C=c, U=u] &= 
\p[ Y=1, M=1 \,\mid\, R=1, D=d, C=c, U=u]\times \p[ R=1 \,\mid\, D=d, C=c, U=u].\label{eq:do.2} 
\end{align}
\end{lemma}

The \eqref{eq:do.2} carries an important implication. On the one hand, the first conditional $\p[ Y=1, M=1 \,\mid\, R=1, D=d, C=c, U=u]$ measures the likelihood of an individual being stopped and subject to law-enforcement actions, given that the incident is reported and the context $U = u$ and criminality $C = c$ are controlled. As such, the first conditional primarily captures the role of law enforcement agents, such as police who are responsible for stops, use of force, and other related actions. On the other hand, the second conditional $\p[ R=1 \,\mid\, D=d, C=c, U=u]$ measures the probability of incidents being reported by the reporting party (e.g. 911 callers), given $U = u$ and $C = c$. This essentially captures the role of the reporting party in law enforcement systems. Together, these two conditionals suggest that racial parity in the entire system is implied by racial parity in each of these stages. This leads to the following definitions.  

\begin{definition}[Racial parity in law enforcement]\label{def:parity.police}
There is {\em racial parity in law enforcement} with respect to $\majority$ and $\minority$ citizens if $\forall u \in \Reals^d$ and $c \in \{0,1\}$, 
\begin{align}
&\p[Y=1, M=1 \,\mid\, R=1, D=\majority, C=c, U=u]
= \p[ Y=1, M=1 \,\mid\, R=1, D=\minority, C=c, U=u]\label{eq:parity.police}.
\end{align}
\end{definition}

\begin{definition}[Racial parity in reporting]\label{def:parity.public}
There is {\em racial parity in reporting} with respect to $\majority$ and $\minority$ citizens if $\forall u \in \Reals^d$ and $c \in \{0,1\}$, 
\begin{align}\label{eq:parity.public}
&\p[ R=1 \,\mid\, D=\majority, C=c, U=u] = \p[ R=1\,\mid\,  D=\minority, C=c, U=u].
\end{align}
\end{definition}
\noindent {Though different notions of fairness can often be incompatible, given the causal DAG in Fig. \ref{fig:systematic_DAG}, parity in public reporting leads to other desirable definitions of fairness like \emph{predictive parity} and \emph{group fairness} \citep{10.1145/3194770.3194776, DBLP:journals/corr/KleinbergMR16, https://doi.org/10.48550/arxiv.1610.07524}. Details of the proof can be found in Appendix Section \ref{appendix:fairness}.} 

The challenge with the above definitions is that criminality $C$ is never fully known. Therefore, it may never be feasible to definitely confirm (or reject) parity in either law enforcement or reporting under this framework. Despite this limitation, next, we discuss certain invariances involving observed quantities ought to be preserved. And, they can lead to test statistics for verifying racial parity. 


\subsection{Verifying Racial Parity With Observed Data}\label{Sec:parity_parity}

The existence of racial parity as defined in Definitions \ref{def:parity.police} and \ref{def:parity.public} are not verifiable using observed quantities or data. However, they do imply certain invariances involving observed quantities only. We shall call them observational parity. 
\begin{definition}[Observational parity in law enforcement]\label{def:observational.parity}
There is {\em observational racial parity in law enforcement} if 
\begin{small}
\begin{align}
\p[Y=1, M=1 \mid R=1, D=\majority, U=u] = \p[ Y=1, M=1 \mid R=1, D=\minority, U=u], ~~\forall u \in \Reals^d \label{eq:observational.parity}. 
\end{align}
\end{small}
\end{definition}
\noindent We similarly define {\em observational parity in law enforcement in context $u$} as the satisfaction of Eq.\eqref{eq:observational.parity} for the given context $u$. This observational parity suggests that a reported majority person and a reported minority person are equally likely to be stopped and later subject to law enforcement actions. We point out that key distinction between Definition \ref{def:parity.police}/\ref{def:parity.public} and observational parity in Definition \ref{def:observational.parity} is that the former consider criminality $C$ while the latter does not. The nuance enables us to directly estimate the probabilities in Eq.\eqref{eq:observational.parity} from observational data without requiring information on criminality. Towards that, we state the following result that connects racial parity (that can not be evaluated) with observational racial parity. The proof can be 
found in Appendix Section \ref{appendix:prop_parity_parity}.
\begin{proposition}\label{prop:parity_parity}
If there is racial parity in both law enforcement and reporting, there is observational parity in law enforcement. 
\end{proposition}

The above result suggests a natural statistical test. Specifically, for each context $u \in \Reals^d$, the \eqref{eq:observational.parity} needs to be satisfied if the racial parity holds. Therefore, violation of it for any context $u \in \Reals^d$ would imply the lack of racial parity. We formally state this next and its proof can be found in Appendix Section \ref{appendix:lemma_parity_parity}.
	
\begin{theorem*}
\label{lemma:parity_parity}
If $\exists u \in \Reals^d$ such that the observational parity in law enforcement in context $u$ is violated, i.e. \eqref{eq:observational.parity} does not hold, then there cannot be parity in law enforcement and reporting simultaneously in the given context $u$. 
\end{theorem*}

With parity in law enforcement and reporting, a reported individual should be equally likely to be stopped and later subject to law enforcement actions regardless of his/her race. In other words, if we find a lack of such observational parity, even {\em without observing} criminality $C$, the true parity considering criminality $C$ in either reporting or law enforcement must be violated. %
Given this, we define the following {\em test statistics}.
\begin{definition}[Observational racial disparity] \label{def.obs.racial.parity}
$\forall u \in \Reals^d$, let {\em observational racial disparity}, denoted as $\Delta(u)$, be 
\begin{align}
	\Delta(u) &:= \p[Y=1, M=1 \,\mid\, R=1, D=\majority, U=u] - \p[ Y=1, M=1 \,\mid\, R=1, D=\minority, U=u]\label{eq:delta.u}.
\end{align} 
\end{definition}
\noindent In effect, $\Delta(u)$ captures the difference in the likelihood of facing a law enforcement action such as use of force between races majority and minority assuming the individual is reported and stopped.

{\color{black} Note that we do not assume the composition of reported incidents $R=1$ to be the same as the composition of unreported incidents $R=0$. The purpose of evaluating \eqref{eq:observational.parity} and \eqref{eq:delta.u}, quantities conditioning on the post-treatment variable $R=1$, is purely for inferring the presence or absence of parity in law enforcement and reporting, which is our real interest. This is vastly different from treating \eqref{eq:observational.parity} as the ultimate fairness definition of interest, which is the problem motivating \cite{KLM}.}

\subsection{Situating Prior Works in Our Framework}\label{ssec:priorworks}

We discuss how prior works, specifically \citep{KLM,Gaebler_Cai_Basse_Shroff_Goel_Hill_2020}, fit within our framework. 

\begin{figure}[h]
    \centering
    \begin{subfigure}[b]{0.26\textwidth}
        \includegraphics[width=\textwidth]{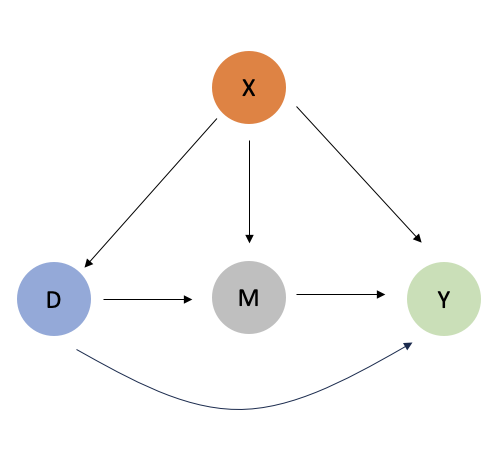}
        \caption{Causal DAG corresponding to prior works \cite{KLM, Gaebler_Cai_Basse_Shroff_Goel_Hill_2020}.}
        \label{fig:prior.dag1}
    \end{subfigure}
    \hfill
    \begin{subfigure}[b]{0.3\textwidth}
        \includegraphics[width=\textwidth]{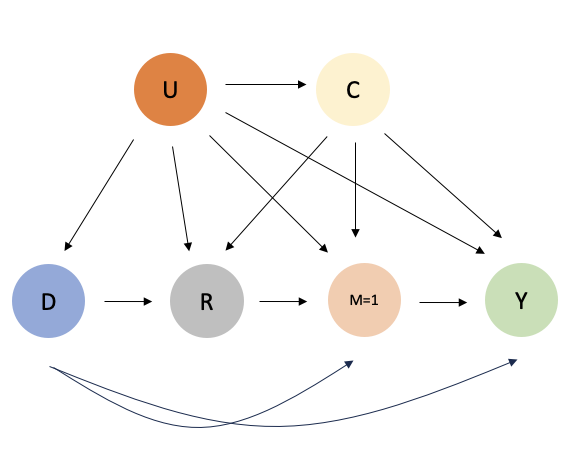}
        \caption{Reduced
    form of Causal DAG of this work.}
        \label{fig:prior.dag2}
    \end{subfigure}
    \hfill
    \begin{subfigure}[b]{0.26\textwidth}
        \includegraphics[width=\textwidth]{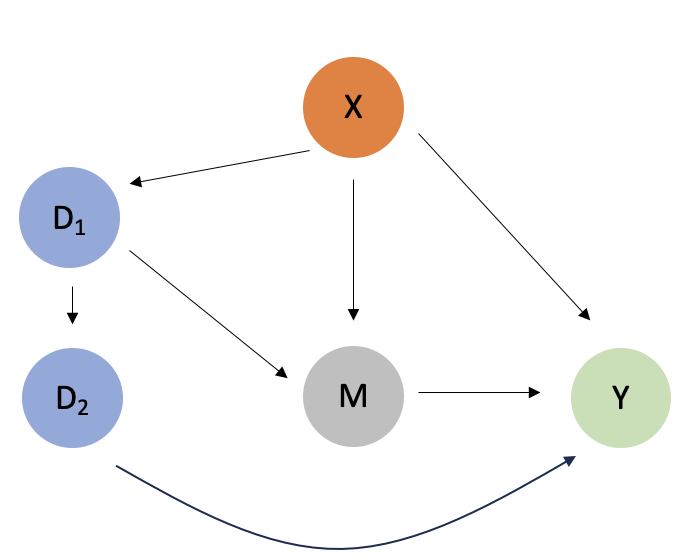}
        \caption{Our interpretation of causal DAG utilized by \cite{KLM} to define counterfactual notion of racial parity.}
        \label{fig:mod.priordag}
    \end{subfigure}
    \caption{Variations of DAGs to situate prior works in our framework.}
    \label{fig:three_images}
\end{figure}

\smallskip
\noindent{\em Comparing Causal DAGs.} As discussed earlier, these prior works have two limitations. First, they do not account for criminality explicitly. Second, accounting for multiple stages, specifically reporting, is absent. The causal DAG in these prior works can be abstracted as the one described in Figure \ref{fig:prior.dag1}. The causal DAG in Figure \ref{fig:prior.dag2} describes the {\em reduced} form of our mechanism where in DAG of Figure \ref{fig:systematic_DAG}, we set $M = 1$ with probability $1$. It can be seen that, with the exception of presence of $C$, the DAGs in Figure \ref{fig:prior.dag1} and Figure \ref{fig:prior.dag2} are identical with mapping $X \leftrightarrow U$ and $M \leftrightarrow R$.

\smallskip
\noindent{\em Comparing Notions of Racial Parity.} We recall the notions of racial parity in \cite{KLM, Gaebler_Cai_Basse_Shroff_Goel_Hill_2020} and compare them with the notion of racial parity of Sections \ref{ssec:parity} and \ref{sec:compare_fairness}. In \cite{KLM, Gaebler_Cai_Basse_Shroff_Goel_Hill_2020}, like in this work, racial parity is defined through a counterfactual notion induced through intervention on race. 

\smallskip
\noindent In \cite{Gaebler_Cai_Basse_Shroff_Goel_Hill_2020}, the racial parity is defined through causal
estimand $\text{CDE}_{\text{M}=1}$ defined as 
\begin{align}
\text{CDE}_{\text{M}=1} &= 
        \mathbb{P}^{do(D=\majority, M=1)}[Y=1|M=1, X=x] - \mathbb{P}^{do(D=\minority, M=1)}[Y=1|M=1, X=x].
\end{align}
Under Causal DAG of Figure \ref{fig:prior.dag1}, per the somewhat controversial ignorability assumption 
$Y^{do(D=d, M=1)} \perp D|X, M=1$ made in \cite{Gaebler_Cai_Basse_Shroff_Goel_Hill_2020}, it follows that 
\begin{align}
    \mathbb{P}^{do(D=d, M=1)}[Y=1|M=1, X=x] & = \mathbb{P}[Y=1|D=d, M=1, X=x].
\end{align}
That is, 
\begin{align}\label{eq:cde.map}
\text{CDE}_{\text{M}=1} &= \mathbb{P}[Y=1|D=\majority, M=1, X=x] - \mathbb{P}[Y=1|D=\minority, M=1, X=x].
\end{align}
Under the mapping described above and explained through Figure \ref{fig:prior.dag1} and Figure \ref{fig:prior.dag2}, it can be checked that \eqref{eq:cde.map} is equivalent to our test statistics in \eqref{eq:delta.u} as defined in Definition \ref{def.obs.racial.parity}.

\smallskip
\noindent Moreover, \cite{KLM} considers another notion of racial parity defined through causal estimand $\text{ATE}_{\text{M}=1}$. To define it, authors introduce a variation of the counterfactual intervention on race. 

While \cite{KLM} do not make this explicit, to explain their approach we introduce a ``modified'' causal DAG as in Figure \ref{fig:mod.priordag} which is obtained by splitting variable $D$ in Figure \ref{fig:prior.dag1} into two variables $D_1, D_2$ with edge between $D_1$ and $D_2$ encoding relationship $D_2 = D_1$. With respect to this DAG, the causal estimand $\text{ATE}_{\text{M}=1}$ can be defined as follows:
\begin{align}
	\text{ATE}_{\text{M}=1} &= \mathbb{P}^{do(D_2=\majority)}[Y=1|M=1, X=x] - \mathbb{P}^{do(D_2=\minority)}[Y=1|M=1, X=x].
\end{align}
While under the original causal DAG of Figure \ref{fig:prior.dag1}, with $D = D_1 = D_2$, the above would have reduced to a similar quantity as in 
$\text{CDE}_{\text{M}=1}$ of \eqref{eq:cde.map}, the modification makes it different. In particular, despite assuming it to be causal DAG, it can not be \textcolor{black}{point-identified} since we do not have two different $D_1, D_2$ in reality. Therefore, \textcolor{black}{\cite{KLM} take a partial identification approach towards $\text{ATE}_{\text{M}=1}$ and derive nonparametric sharp bounds}. 

\smallskip
\noindent{\em In Summary.} The two important prior works \cite{KLM, Gaebler_Cai_Basse_Shroff_Goel_Hill_2020} can be explained through the
 reduced form of our framework. In the reduced form to match the prior work, our test statistics matches with that of \cite{Gaebler_Cai_Basse_Shroff_Goel_Hill_2020}. In that sense, ours can be viewed as a natural generalization of \cite{Gaebler_Cai_Basse_Shroff_Goel_Hill_2020} to account for criminality as well as multi-stage effect. The causal estimand of \cite{KLM} can also be explained within the same framework, however it requires \textcolor{black}{additional assumptions and data} to lead to \textcolor{black}{point-identification} using observed data.

\section{Measuring Racial Disparity: Law Enforcement, Reporting}\label{sec:racial_disparity}

The Theorem \ref{lemma:parity_parity} of Section \ref{sec:systemic} provides test statistics to verify whether there is racial parity or not. Specifically, 
if for any context $u \in \Reals^d$, if $\Delta(u)$ as defined in \eqref{eq:delta.u} is $\neq 0$, then racial disparity does exist. However, it does not help determine the primary source of disparity nor its intensity. 

Indeed, in general, it is infeasible to pinpoint the primary source and intensity of disparity given that not everything is observed. However, in this section, we shall discuss three canonical scenarios where it would be feasible to determine the primary source of disparity and intensity of racial disparity. 

To that end, we define the proportion of innocents ($C=0$) being reported ($R=1$) for a given context $u \in \Reals^d$ and race $d \in \{\minority, \majority\}$ 
as
\begin{align}\label{eq:prop.innoc}
    \xi_{d,u} & = \p[C=0\,\mid\, R=1,  D=d, U=u].
\end{align}
This quantity suggests three canonical scenarios: 

{\em Scenario 1 (aka Airport security checks): $\forall d \in \{\majority, \minority\}, \forall u \in \Reals^d, \xi_{d, u} \simeq 1  $,} i.e., most reported individuals are innocent. This naturally happens in settings such as airport security check with everyone coming to the airport voluntarily 
submits to security check (i.e. $R=1$). And indeed, most (if not all) are likely to be innocent.   

{\em Scenario 2 (aka Advanced AI-empowered policing): $\forall d \in \{\majority, \minority\}, \forall u \in \Reals^d, \xi_{d, u} \simeq 0$, } i.e., most reported individuals are criminals. This naturally arises in settings such as (not so futuristic) AI-empowered policing for various forms. Ideally, a good AI solution would highlight those incidents that are likely to have higher chance of crime. That is, if a good AI-system reports an 
incident (i.e. $R = 1$), then the chance of crime is likely to be higher. However, such usage of AI solutions in policing should always be cautioned and heavily monitored.

{\em Scenario 3 (aka Police-civilian interactions): $\forall d \in \{\majority, \minority\}, \forall u \in \Reals^d, 0 < \xi_{d, u} < 1$, } i.e., a reasonable fraction of reported individuals are criminals as well as innocents. This naturally arises in the settings such as police-civilian interaction, including traffic stops and stop-frisk. Specifically, a good fraction of incidences reported through 911 calls involve innocents. 

In what follows, we discuss how in each of these scenarios, under natural conditions how the primary source of racial disparity can be identified if $\Delta(u) \neq 0$ for some $u \in \Reals^d$. Towards that, we derive a characterization that would be very useful.

\subsection{Observational Racial Disparity: A Useful Characterization}

We present an alternate characterization of $\Delta(u)$ for $u \in \Reals^d$ involving $\xi_{d, u}$ with $d \in \{\majority, \minority\}$. Towards that, we define the following: for $u \in \Reals^d$ and $d \in \{\majority, \minority\}$, 
\begin{align}
Y_{\criminal, d, u} & = \p[ Y=1, M=1 \,\mid\, R=1, C=1, D=d, U=u], \nonumber \\ 
Y_{\innocent, d, u} & = \p[ Y=1, M=1 \,\mid\, R=1, C=0, D=d, U=u].
\end{align}
The above represent the probability of criminals and innocents, respectively, being subject to law enforcement actions.
\begin{theorem*}\label{thm: delta_representation}
For any $u\in\Reals^d$, the observational racial disparity in law enforcement can be represented as 
\begin{align}
	\label{eq:delta.u.char}
\Delta(u) &=
	    \Big (Y_{\innocent, \majo, u}\cdot \xi_{\majo, u}+Y_{\criminal, \majo, u}\cdot (1-\xi_{\majo, u})\Big )-\Big ({Y_{\innocent, \mino, u}\cdot\xi_{\mino, u}} +Y_{\criminal, \mino, u}\cdot (1-\xi_{\mino, u})\Big ).
\end{align}
\end{theorem*}
The proof of Theorem \ref{thm: delta_representation} can be found in Appendix Section \ref{appendix:delta_representation_proof}.

\subsection{Scenario 1 (aka Airport security checks)}\label{sec:mostly_innocent}

As discussed earlier, in the setting such as airport security checks, where individuals self-report themselves and are later checked by TSA officers, 
it is reasonable to assume that most people passing through these security checks are innocent, not criminals. In our notation, this suggests that no matter what the context $u$ and race $d$ is, $\xi_{d, u} \simeq 1$. Therefore, by characterization of \eqref{eq:delta.u.char}, we have $ \Delta(u) \simeq Y_{\innocent, \majo, u}-Y_{\innocent, \mino, u}.$ Formally, we state the following.
\begin{proposition}\label{prop:scenario.1}
For each $d \in \{\majority, \minority\}$ and $u \in \Reals^d$, let $\xi_{d, u}\to 1$. Then, 
\begin{align}
\Delta(u) < 0 & \Leftrightarrow  Y_{\innocent, \mino, u} > Y_{\innocent, \majo u}. 
\end{align}
\end{proposition}
Proof can be found in Appendix Section \ref{appendix:mostly_innocent}. \\

\noindent {\em Key takeaway:} in scenario 1, if there is observational disparity against the minority (resp. majority) in law enforcement, it is primarily due to bias in law enforcement actions against the minority (resp. majority) innocents. That is, $\Delta(u)<0 \Leftrightarrow Y_{\innocent, \mino, u} > Y_{\innocent, \majo u}.$  

\subsection{Scenario 2 (aka Advanced AI-empowered policing)}\label{sec:mostly_guilty}

Consider a (not so futuristic) scenario where the reporting system is enhanced with physical sensors and detectors that use accurate artificial intelligence (AI) \footnote{\textcolor{black}{We do not intend to advocate for the implementation of AI technologies in law enforcement. Also, we refer to ``AI'' in a broader, more generalizable context. For example, it's akin to using automated systems like motion-sensor lights that automatically turn off after five minutes of inactivity in unoccupied conference rooms.}}. For example, an advanced gunshot detection system can use AI to recognize a blast received by the acoustic sensors as a gunshot and then pinpoint its exact location to alert the police about this potential gun-related violation or crime. Similarly, a face recognition system that implements Computer Vision algorithms to recognize wanted criminals from the street surveillance footage can report them to the police. With the help of AI, these systems often accurately identify criminals, especially those who have committed felonies. In our notation, this suggests that no matter what the context $u$ and race $d$ is, $\xi_{d, u} \simeq 0$. Therefore, by characterization of \eqref{eq:delta.u.char}, we should have $\Delta(u) \simeq Y_{\criminal, \majo, u}-Y_{\criminal, \mino, u}.$ Formally, we state the following.
\begin{proposition}\label{prop:scenario.2}
For each $d \in \{\majority, \minority\}$ and $u \in \Reals^d$, let $\xi_{d, u}\to 0$. Then, 
\begin{align}
\Delta(u)<0 & \Leftrightarrow Y_{\criminal, \mino, u} > Y_{\criminal, \majo u}. 
\end{align}
\end{proposition}
Proof can be found in Appendix Section \ref{appendix:mostly_guilty}.\\

\noindent{\em Key takeaway:} in scenario 2, if there is observational disparity against the minority (resp. majority) in law enforcement, it is primarily due to bias in law enforcement actions against minority (resp. majority) criminals.  

\subsection{Scenario 3 (aka Everyday police-civilian interactions)}\label{sec:more_or_less}

Now consider a scenario, such as police-civilian interactions, where the reporting is initiated through 911 calls. In this case, a substantial fraction of reporting corresponds to some form of criminal activity, but there is also a non-trivial amount of human error. That is, $\xi_{d, u}$ is neither close to $0$, nor close to $1$. Understanding the source of observational disparity can become rather challenging in this setting. There are three factors to consider: (a) the accuracy of the public 911 reporting, (b) policing actions on criminals, and (c) policing actions on innocents. Imposing restrictions on any two of these three factors can help us understand the other factor's impact in inducing observational disparity. We do that next. 

\noindent{\bf Case 1.} Consider the setting where public reporting is unbiased and policing actions against criminals are unbiased. That is, the 911-reported civilians of different races are equally likely to be innocents (${\xi_{\majo, u}} \simeq {\xi_{\mino, u}}$), and criminals of different races are equally likely to be subject to policing actions (${Y_{\criminal, \majo, u}}\simeq {Y_{\criminal, \mino, u}}$). Therefore, by characterization of \eqref{eq:delta.u.char}, we have $\Delta(u) \simeq \xi_{\majo, u}(Y_{\innocent, \majo, u}-Y_{\innocent, \mino, u})$. Formally, we state the following.
\begin{proposition}\label{prop:scenario.3.1}
For each $d \in \{\majority, \minority\}$ and $u \in \Reals^d$, let ${\xi_{\majo, u}} - {\xi_{\mino, u}} \to 0$ and 
${Y_{\criminal, \majo, u}} - {Y_{\criminal, \mino, u}} \to 0$. Then, 
\begin{align}
\Delta(u) & \to \xi_{\majo, u}(Y_{\innocent, \majo, u}-Y_{\innocent, \mino, u}). 
\end{align}
\end{proposition}
Proof can be found in Appendix Section \ref{appendix:proof3.1}.\\

\noindent{\em Key takeaway:} in scenario 3 case 1, there is observational disparity against the minority in law enforcement if and only if the minority innocents are more often subject to policing actions. \\

\noindent{\bf Case 2.} Consider the setting where public reporting is unbiased and policing action for innocents is unbiased. 
That is, the 911-reported civilians of different races are roughly equally likely to be innocents (${\xi_{\majo, u}} \simeq {\xi_{\mino, u}}$), and innocents of different races are equally likely to be subjected to policing actions (${Y_{\innocent, \majo, u}}\simeq {Y_{\innocent, \mino, u}}$).  Therefore, by characterization of \eqref{eq:delta.u.char}, we have
$\Delta(u) \simeq (1-\xi_{\majo, u}) (Y_{\criminal, \majo, u} - Y_{\criminal, \mino, u}).$
Formally, we state the following.
\begin{proposition}\label{prop:scenario.3.2}
For each $d \in \{\majority, \minority\}$ and $u \in \Reals^d$, let ${\xi_{\majo, u}} - {\xi_{\mino, u}} \to 0$ and 
${Y_{\innocent, \majo, u}} - {Y_{\innocent, \mino, u}} \to 0$. Then, 
\begin{align}
\Delta(u) & \to (1-\xi_{\majo, u}) (Y_{\criminal, \majo, u} - Y_{\criminal, \mino, u}). 
\end{align}
\end{proposition}
Proof can be found in Appendix Section \ref{appendix:proof3.2}. \\

\noindent{\em Key takeaway:} in scenario 3 case 2, there is observational disparity against the minority in law enforcement if and only if the minority criminals are more  often subject to policing actions. \\
  
\noindent{\bf Case 3.} Consider the setting where policing action for innocents and criminals is unbiased. That is, innocents and criminals of different races are equally likely to be subjected to policing actions (${Y_{\innocent, \majo, u}}\simeq {Y_{\innocent, \mino, u}}$ and ${Y_{\criminal, \majo, u}}\simeq {Y_{\criminal, \mino, u}}$). 
Therefore, by characterization of \eqref{eq:delta.u.char}, we have $\Delta(u) \simeq (Y_{\criminal, \majo, u}-Y_{\innocent, \majo, u})(\xi_{\mino, u}-\xi_{\majo, u}).$
Formally, we state the following.
\begin{proposition}\label{prop:scenario.3.3}
For $d \in \{\majority, \minority\}$ $u \in \Reals^d$, let ${Y_{\innocent, \majo, u}} - {Y_{\innocent, \mino, u}} \to 0$, ${Y_{\criminal, \majo, u}} - {Y_{\criminal, \mino, u}} \to 0$. Then, 
\begin{align}
\Delta(u) & \to (Y_{\criminal, \majo, u}-Y_{\innocent, \majo, u})(\xi_{\mino, u}-\xi_{\majo, u}). 
\end{align}
\end{proposition}
Proof can be found in Appendix Section \ref{appendix:proof3.3}. \\

\noindent{\em Key takeaway:} in scenario 3 case 3, there is observational disparity against the minority in law enforcement if and only if the 911-reported majorities are more likely to be innocents. That is, observational bias against the minority is due to actual reporting bias against the majority race!

\section{Empirics}\label{sec:empirics}

In this section, we utilize publicly available datasets on calls for service (i.e. 911 calls) and police stops from NYC and New Orleans to evaluate the implications of our framework. At its core, our goal is to understand whether there is evidence of bias either in law enforcement or reporting or both given the observed data. To that end, Section \ref{sec:systemic} provides a data-driven way to evaluate racial disparity. In particular, Proposition  \ref{prop:parity_parity} as well Propositions \ref{prop:scenario.3.1}-\ref{prop:scenario.3.3} provide means to interpret the observational racial disparity. \\

\noindent {\em Test Statistics and Interpretation.} Given context $u \in \Reals^d$, the test statistics is $\Delta(u)$, the observational disparity in law enforcement. Using available data, for each $u \in \Reals^d$, we test the hypotheses 
\begin{align}\label{eq:null.hyp}
    H_0(u) & : \Delta(u) = 0 \text{ vs. } H_1(u) : \Delta(u) \neq 0
\end{align} 
Specifically, we compute the 95\% confidence interval for the value of $\Delta(u)$, and if the interval does not contain $0$, then we reject $H_0(u)$. In that case, per Proposition  \ref{prop:parity_parity}, we conclude presence of racial disparity (of some form). 

If we conclude presence of racial disparity, to pin-point the cause further, we utilize Propositions \ref{prop:scenario.3.1}-\ref{prop:scenario.3.3}. 
\begin{itemize}
    
    \item {\em Case 1. } If public reporting and policing actions with respect to criminals are unbiased, then
    \begin{itemize}
        \item $\Delta(u) < 0$ is primarily due to the bias against the minority innocents in policing actions.  
        \item $\Delta(u) > 0$ is primarily due to the bias against the majority innocents in policing actions. 
    \end{itemize}

    \item {\em Case 2. } If public reporting and policing action with respect to innocents are unbiased, then
    \begin{itemize}
        \item $\Delta(u) < 0$ is primarily due to the bias against minority criminals in policing actions.  
        \item $\Delta(u) > 0$ is primarily due to the bias against majority criminals in policing actions. 
    \end{itemize}
    \item {\em Case 3. } If policing action with respect to innocents and criminals are unbiased, then
    \begin{itemize}
        \item $\Delta(u) < 0$ is primarily due to the bias against the majority in public reporting.  
        \item $\Delta(u) > 0$ is primarily due to the bias against the minority in public reporting. 
    \end{itemize}
\end{itemize}
Specifically, given the context, we shall utilize one of the three cases above to conclude appropriate likely primary cause for racial disparity. \\

\noindent {\em Details of Empirical Study.} Various details related to our empirical study are provided in a thorough manner in Appendix Section \ref{appendix:seattle} due to space constraints. However, we highlight a few key aspects here. First, to define the {\em context} (i.e. $U$), we utilize the precinct as a proxy. This is because a precinct often (1) consists of a comparatively homogeneous neighborhood in terms of the socioeconomic status and (2) is served by the same batch of police officers for relatively consistent law enforcement. The details about data and associated necessary data processing related to the 911 call, police stop and associated use of force are explained in Appendix Section \ref{appendix:seattle}. Next we describe our findings for NYC and New Orleans, and our likely conclusions.


\begin{figure}[htbp]
  \centering
  \includegraphics[width=1.01\linewidth]{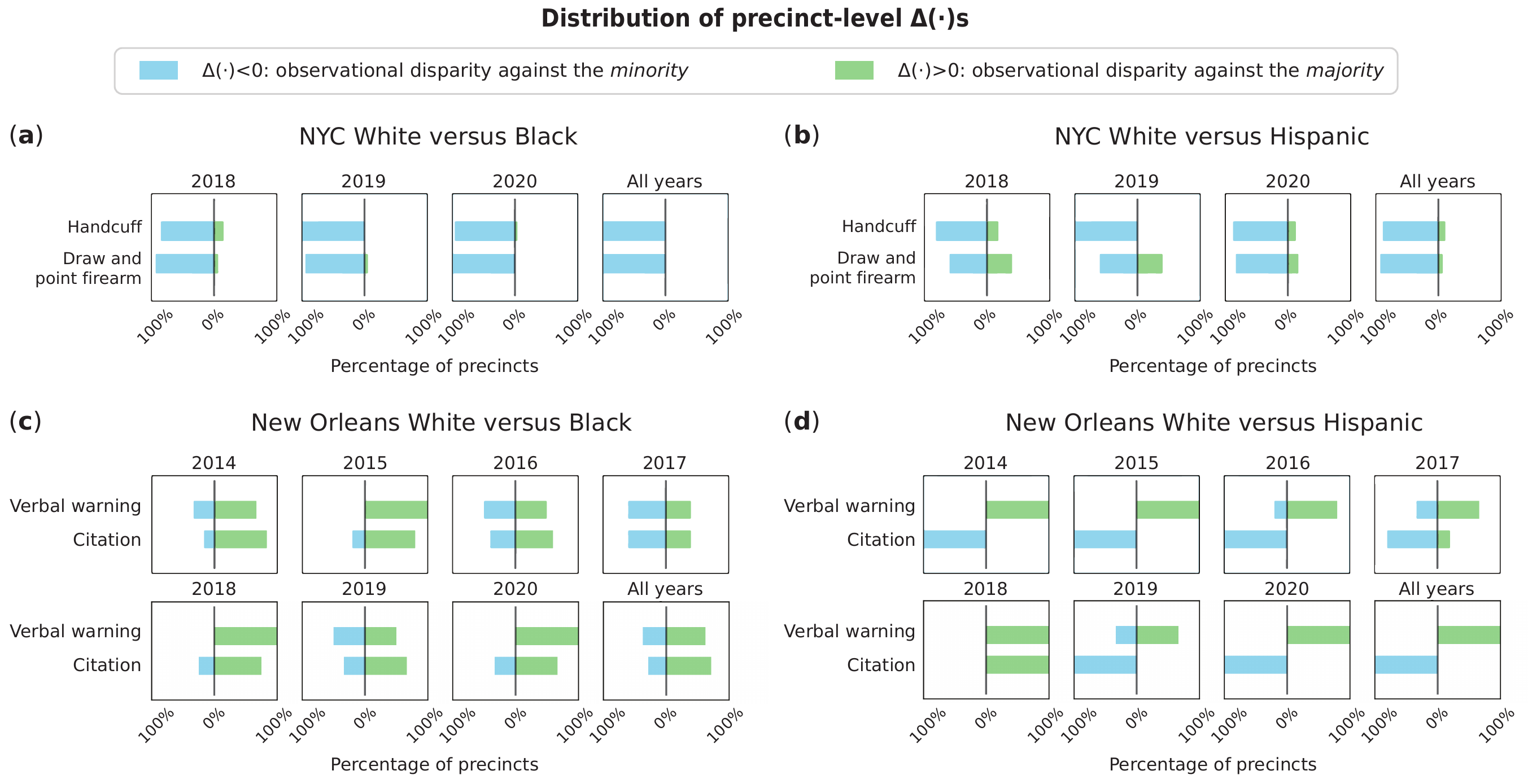}
  \caption{Distribution of statistically significant precinct-level $\Delta(\cdot)$s for different policing actions. $\Delta(\cdot)<0$ indicates bias against minority (Black/Hispanic); $\Delta(\cdot)>0$ indicates bias against majority (White). (a) NYC White versus Black; (b) NYC White versus Hispanic; (c) New Orleans White versus Black; (d) New Orleans White versus Hispanic. In summary, there is predominantly observation bias against minority (Black / Hispanic) compared to majority (White) in NYC while in New Orleans there is predominantly observation bias against  majority (White) compared to minority (Black / Hispanic). We conclude that the primary source of bias in NYC is policing action against minority, while in New Orleans it is in the 911 calling (reporting) against minority!} 
  \label{fig:both}
\end{figure}

\subsection{New York City (NYC)}

{\bf Data.} NYC police commands exercise their authority over 77 police precincts. We utilize the 911 call data and pedestrian stop data over the years 2018, 2019, and 2020. We consider {\em White} as the majority race while {\em Black} or {\em Hispanic} as the minority race in the data analysis. 
After appropriate data processing, it contains $19,509,773$ total 911 calls-for-service, $34,009$ stops among which $21,130$ are induced from 911 calls. In this sense, $0.11\%$ of 911 calls result in stops and $62.13\%$ are induced by 911 calls. This suggests that only a small portion of 911 calls end up in stops but the majority of stops are from 911 calls. It stresses the importance of 911 calls in police-civilian encounters. The detailed year-to-year data statistics are summarized in Table~\ref{table:stats_nyc} in Appendix Section \ref{appendix:statistics}. 

\noindent In keeping with \cite{Fryer2019}, use of force in a police-civilian interaction is classified at different levels with an additional level $0$ corresponding to {\em no use of force}.  A police-civilian interaction may involve several different levels of use of force and we associate the incident with the highest level of force used.  \\

\noindent {\bf Key Takeaway.} We obtain the results of $\Delta(\cdot)$s  for NYC across all precincts as summarized in Figure \ref{fig:both} (a) and (b). The blue bars and green bars represent the proportions of precincts with statistically significant negative $\Delta(\cdot)$s and positive $\Delta(\cdot)$s, i.e. reject the null hypothesis \eqref{eq:null.hyp}. 

We see that across different levels of use of force, it is comparatively consistent that most statistically significant $\Delta(\cdot)$'s are negative, exhibiting observational bias in law enforcement against the minority. In the case of NYC, it seems that Case 1 or Case 2 discussed above is likely. That is, the reason for observed racial disparity against minority is primarily due to {\em bias in policing action against innocents or criminals of the minority race}.

\subsection{New Orleans}

\noindent{\bf Data.} New Orleans is divided into eight police districts.  After data pre-processing, from 2014 to 2020, there are $2,160,679$ total 911 calls-for-service and  $349,942$ stops. Of these stops, $134,746$ are induced from 911 calls. That is, $6.24\%$ of 911 calls result in stops and $38.51\%$ of stops are induced from by 911 calls. The detailed year-to-year statistics of data is summarized in Table~\ref{table:stats_no} in Appendix Section \ref{appendix:statistics}.

\noindent{\bf Key Takeaway.} The summary findings for New Orleans are shown in Figure \ref{fig:both} (c) and (d). The blue bars and green bars represent the proportions of precincts with statistically significant negative $\Delta(\cdot)$s and positive $\Delta(\cdot)$s. As can be seen, unlike NYC, for New Orleans, we consistently find most statistically significant $\Delta(\cdot)$'s are positive, exhibiting observational bias in law enforcement {\em against the majority. }

Now if we consider Cases 1 or 2, then it would conclude that the police action is biased against majority which is not a likely scenario. Given that, the Case 3 is more likely. In which case, we would conclude that primary source for {\em observational bias against majority is the bias against minorities (through over-reporting) in 911 calls-for-service.}


\section{Conclusion and Future Work}\label{sec:discussion}

In this work, we have developed a causal framework to evaluate racial bias or disparity in a multi-stage law enforcement system. We furbished it with data-driven test to detect the presence of bias. We identified three canonical scenarios where it is feasible to identify the primary source of bias when it is present. We applied our framework to the setting of police-civilian interaction using police stop data and 911 call data from NYC and New Orleans. We found that there is observational racial disparity against the minority in NYC while against the majority in New Orleans. Our framework leads to the conclusion that the likely source of bias in NYC is the biased policing actions against the minority while the likely source of bias in New Orleans is the biased 911 call reporting against the minority.  

We highlight that discerning the potential source of racial bias in law enforcement systems is not to assign blame but to inform effective policy reforms. It is important to note that our conclusions are derived from empirical analyses using data collected over a certain period of time. Therefore, we advocate that any proposed policy changes should be grounded with the latest data to ensure that reforms are relevant, timely, and based on the most current understanding of the issues at hand. 

Our work advances state-of-art in multiple ways. First, it provides a comprehensive framework by including notion of criminality leading to a proper notion racial (dis)parity, and thus overcoming the seeming debate in the prior works of \cite{KLM, Gaebler_Cai_Basse_Shroff_Goel_Hill_2020} over the right definition of causal estimand to determine the racial (dis)parity. Second, it provides a multi-stage framework to capture the entire causal chain of interactions. Third, our work provides data-driven test to detect the racial disparity. Fourth, through theoretical analysis, we identify canonical scenarios where it is feasible to identify the primary source of bias in the presence of it. Finally, the empirical analysis using police stop data along with 911 calls data provides new insights into the study of racial bias that has received extensive attention in the recent literature.

Our framework has applicability beyond the police-civilian interaction. For example, its applicability in the context of AI-empowered policing with sharp theoretical understanding would be extremely useful for evaluating the biases within AI-empowered systems as they get deployed more and more in the upcoming future. For instance, a similar causal framework can be utilized to evaluate bias within AI-empowered recruitment process. In particular, a recruitment process is a multi-stage system where a candidate of a particular race must first have access to a job opening information to apply, be interviewed, and then either be hired or not. Minority candidates with qualifications may already be at a disadvantage compared to majority candidates even before the actual interviews if they have less access to job opening information. And if an AI-empowered system is utilized to decide on whether to interview a candidate or not based on their resumes and job application, then it starts looking very similar to the scenario of AI-empowered policing we discussed earlier.

\bibliographystyle{ACM-Reference-Format}
\bibliography{sample-base}

{\color{black}

}

\bigskip


\appendix

\newpage
\section{Proof of Lemma \ref{lemma:calculation}}\label{appendix:prop_do_calculus}

\begin{proof}
Using the law of total probability, we have
\begin{align}
& \p^{\doop{D=d}}[ Y=1 \,\mid\, C=c, U=u]  \\
=& \sum_{r, m \in \{0,1\}} \p^{\doop{D=d}}[ Y=1, M=m, R=r,  D = d  \,\mid\, C=c, U=u] \nonumber \\
= &\sum_{r, m \in \{0,1\}} \p^{\doop{D=d}}[ Y=1, M=m \,\mid\, R=r, D=d, C=c, U=u] \nonumber \\ 
& \qquad\qquad\qquad \times \p^{\doop{D=d}}[ R=r \,\mid\, D=d, C=c, U=u] \times \p^{\doop{D=d}}[ D=d \,\mid\, C=c, U=u] \label{eq:do.1} \\
=&\p[ Y=1, M=1 \,\mid\, R=1, D=d, C=c, U=u] \times \p[ R=1 \,\mid\, D=d, C=c, U=u]. \label{eq:do.2.2}
\end{align}
where \eqref{eq:do.2.2} is obtained from \eqref{eq:do.1} since $Y = 1$ only if $M = 1$ and $R=1$, and by do-calculus and the causal DAG, $\p^{\doop{D=d}}[ D=d \,\mid\, C=c, U=u]=1$ and $\p^{\doop{D=d}}(\cdot \,\mid\, \cdot) = \p(\cdot\,\mid\, \cdot)$, i.e., the statistical conditionals. 
\end{proof}

{\color{black} \section{Compatibility among some fairness notions}\label{appendix:fairness}
\begin{lemma}
    Let \emph{predictive parity} be $\p(C = 1| R = 1, D = \majority, U=u) = \p(C = 1| R = 1, D = \minority, U=u)$ and \emph{group fairness} be $\p(R = 1|D = \majority, U=u) = \p(R = 1|D = \minority, U=u)$. Given the causal DAG in Fig. \ref{fig:systematic_DAG}, parity in reporting leads to \emph{predictive parity} and \emph{group fairness}.
\end{lemma}

\begin{proof}
    First, we want to show given 
    \begin{align}
        \p[ R=1 \,\mid\, D=\majority, C=c, U=u] &= \p[ R=1\,\mid\,  D=\minority, C=c, U=u] \text{\ and}\label{eq:rr}\\
        \p[ C=c \,\mid\, D=\majority, U=u] &= \p[ C=c\,\mid\,  D=\minority, U=u], \label{eq:cc}
    \end{align}
we have 
\begin{align}
    \p[ C=c \,\mid\, D=\majority, R=1, U=u] = \p[ C=c\,\mid\,  D=\minority, R=1, U=u]\label{eq:cd}.
\end{align}

    First, notice that 
\begin{align*}
    &\frac{\p[C=1\,\mid\, D=d, R=1, U=u]}{1-\p[C=1\,\mid\,D=d, R=1, U=u]} \\
    =&  \frac{\p[C=1\,\mid\, D=d, U=u]}{\p[C=0\,\mid\, D=d, U=u]}\cdot \frac{\p[ R=1 \,\mid\, D=d,C=1, U=u]/\p[R=1\,\mid\, D=d, U=u]}{\p[ R=1 \,\mid\, D=d, C=0, U=u]/\p[R=1\,\mid\, D=d, U=u]}\\
    =& \frac{\p[C=1\,\mid\, U=u]}{\p[C=0\,\mid\, U=u]}\cdot \frac{\p[ R=1 \,\mid\, C=1, U=u]}{\p[ R=1 \,\mid\, C=0, U=u]}.
\end{align*}
Thus 

\begin{align*}
    &\frac{\p[C=1\,\mid\, D=\majority, R=1, U=u]}{1-\p[C=1\,\mid\,D=\majority, R=1, U=u]} = \frac{\p[C=1\,\mid\, D=\minority, R=1, U=u]}{1-\p[C=1\,\mid\,D=\minority, R=1, U=u]}\\
    \Leftrightarrow& \p[C=c\,\mid\, D=\majority, R=1, U=u] = \p[C=c\,\mid\, D=\minority, R=1, U=u].
\end{align*}

Then is straightforward given \eqref{eq:rr}, \eqref{eq:cc} and \eqref{eq:cd}, we have
\begin{align}
    \p[R=1\,\mid\,D=\majority, U=u] &= \p[R=1\,\mid\,D=\minority, U=u]\label{eq:rrr1}.
\end{align}


In fact, any two equalities of \eqref{eq:rr}, \eqref{eq:cc}, \eqref{eq:cd}, and \eqref{eq:rrr1}, except \eqref{eq:cd} and \eqref{eq:rrr1}, can lead to the other two equalities among the four. 
\end{proof}
}

\section{Proof of Proposition \ref{prop:parity_parity}}\label{appendix:prop_parity_parity}

\begin{proof}
Define
\begin{align*}
	R_{\criminal, \majo, u} &= \p[R=1 \,\mid\, C=1, D=\majority, U=u], \\
   	R_{\innocent, \majo, u} &= \p[R=1 \,\mid\, C=0, D=\majority, U=u], \\
    C_{\criminal, \majo, u} &= \p[ C=1 \,\mid\, D=\majority, U=u], \\
    C_{\innocent, \majo, u} &= \p[ C=0 \,\mid\, D=\majority, U=u].
\end{align*}
Given the notations in Section \ref{sec:racial_disparity}, the proposition is thus simplified to $\forall u \in \mathbb{R}^d$, if $Y_{\criminal, \majo, u} = Y_{\criminal, \mino, u}, Y_{\innocent, \majo, u} = Y_{\innocent, \mino, u}, R_{\criminal, \majo, u} = R_{\criminal, \mino, u}, \text{ and }R_{\innocent, \majo, u} = R_{\innocent, \mino, u}$, then we should have $\p[ Y=1, M=1 \,\mid\, R=1, D=\majority, U=u]=\p[ Y=1, M=1 \,\mid\, R=1, D=\minority, U=u]$. 

Noticing that the if statement considers whether an individual is a criminal or not while the then statement does not, we thus hope to use the quantities involving whether a person has committed the crime to represent the quantities not involving so. Thus, we have $\forall d \in \{\majority, \minority\}$ and $\forall u \in \mathbb{R}^d$,
{\footnotesize \begin{align}
&\p[ Y=1, M=1 \,\mid\, R=1, D=d, U=u] \nonumber\\
=&\frac{\p[ Y=1, M=1, R=1 \,\mid\,  D=d, U=u]}{\p[ R=1 \,\mid\,  D=d, U=u]}\tag{By definition of conditional probability}\\
=& \frac{\sum\limits_{c \in \{0,1\}} \p[ Y=1, M=1, R=1, C=c \,\mid\,  D=d, U=u]}{\sum\limits_{c \in \{0,1\}} \p[ R = 1, C=c \,\mid\, D=d, U=u]} \tag{By law of total probability}\\
 =& \frac{\sum\limits_{c \in \{0,1\}} \p[ Y=1, M=1 \,\mid\,  R=1, C=c, D=d, U=u] \p[R=1 \,\mid\, C=c, D=d, U=u] \p[ C=c \,\mid\,  D=d, U=u]}{\sum\limits_{c \in \{0,1\}} \p[R=1 \,\mid\, C=c, D=d, U=u] \p[ C=c \,\mid\,  D=d, U=u]} \tag{By chain rule of conditional probability on both the numerator and denominator}
\\
=&\begin{cases}\label{eq:comparison}
   \dfrac{Y_{\innocent, \majo, u}R_{\innocent, \majo, u}C_{\innocent, \majo, u}+Y_{\criminal, \majo, u}R_{\criminal, \majo, u}C_{\criminal, \majo, u}}{R_{\innocent, \majo, u}C_{\innocent, \majo, u}+R_{\criminal, \majo, u}C_{\criminal, \majo, u}} \text{ if } d=\majority \\
   \dfrac{Y_{\innocent, \mino, u}R_{\innocent, \mino, u}C_{\innocent, \mino, u}+Y_{\criminal, \mino, u}R_{\criminal, \mino, u}C_{\criminal, \mino, u}}{R_{\innocent, \mino, u}C_{\innocent, \mino, u}+R_{\criminal, \mino, u}C_{\criminal, \mino, u}}\text{ if } d=\minority.\\  
   \end{cases}
\end{align}}%

We now want to show the quantity for majority and the one for minority in (\ref{eq:comparison}) equal each other. From above, we know we are given $Y_{\criminal, \majo, u} = Y_{\criminal, \mino, u}$ and $Y_{\innocent, \majo, u} = Y_{\innocent, \mino, u}$ due to parity in policing and $R_{\criminal, \majo, u} = R_{\criminal, \mino, u}$ and $R_{\innocent, \majo, u} = R_{\innocent, \mino, u}$ due to parity in public reporting. Moreover, given the systematic causal DAG, we know $C$ is independent of $D$ given $U$, i.e. $C_{\innocent, \majo, u} = C_{\innocent, \mino, u}$ and $C_{\criminal, \majo, u} = C_{\criminal, \mino, u}$. With the six equalities listed above, we have the quantity for the majority the same as for the minority in (\ref{eq:comparison}). Therefore, given parity in policing and public reporting, we have the observational parity $\p[ Y=1, M=1 \,\mid\, R=1, D=\mino, U=u]=\p[ Y=1, M=1 \,\mid\, R=1, D=\majo, U=u]$ $\forall u \in \mathbb{R}^d$ which can be obtained from real-world data.
\end{proof}

\section{Proof of Theorem \ref{lemma:parity_parity}}\label{appendix:lemma_parity_parity}
\begin{proof}
Notice that this theorem is exactly the contraposition of Proposition \ref{prop:parity_parity}, so it is proved automatically. With this theorem, if under some context $U=u$ we notice $\p[ Y=1, M=1 \,\mid\, R=1, D=\majo, U=u]\neq\p[ Y=1, M=1 \,\mid\, R=1, D=\mino, U=u]$, then we know immediately at least one of the parities (parity in public reporting and parity in policing) must be violated. 
\end{proof}

\section{Proof of Theorem \ref{thm: delta_representation}}\label{appendix:delta_representation_proof}
\begin{proof}
Using the same technique in Appendix Section \ref{appendix:prop_parity_parity}, we first represent the two observational quantities with the quantities involving whether a person has committed the crime or not. $\forall c \in \{0, 1\}, d\in \{\majority, \minority\}$ and $u \in \Reals^d$,
\begin{align}
& \p[ Y=1, M=1 \,\mid\, R=1, D=d, U=u]\nonumber\\
=& \sum_{c\in\{0, 1\}}\p[ Y=1, M=1, C=c\,\mid\, R=1, D=d, U=u]\tag{By law of total probability}\\
=& \sum_{c\in\{0, 1\}}\p[ Y=1, M=1\,\mid\, C=c, R=1, D=d, U=u] \cdot \p[ C=c \,\mid\, R=1, D=d, U=u]\tag{By chain rule}\\
=& Y_{\innocent, d, u} \cdot \xi_{d, u}+Y_{\criminal, d, u} \cdot (1-\xi_{d, u})\label{eq:single_representation}
\end{align}
Thus, plugging Eq.\eqref{eq:single_representation} into $\Delta(u)$, we have
\begin{align}
\Delta(u)=&\p[ Y=1, M=1 \,\mid\, R=1, D=\text{majority}, U=u] \\
&- \p[ Y=1, M=1 \,\mid\, R=1, D=\text{minority}, U=u]\nonumber\\
=&\Big (Y_{\innocent, \majo, u} \cdot \xi_{\majo, u}+Y_{\criminal, \majo, u} \cdot (1-\xi_{\majo, u})\Big )\\
&-\Big (Y_{\innocent, \mino, u} \cdot \xi_{\mino, u}+Y_{\criminal, \mino, u} \cdot (1-\xi_{\mino, u})\Big )\nonumber
\end{align}
Note that $Y_{\innocent, d, u} \cdot \xi_{d, u} = \p[ Y=1, M=1, C=0\,\mid\, R=1, D=d, U=u]$ and $Y_{\criminal, d, u} \cdot (1-\xi_{d, u}) = \p[ Y=1, M=1, C=0\,\mid\, R=1, D=d, U=u]$, we interpret the first term in each pair of large parentheses as the likelihood of a reported individual being an innocent subject to policing actions and the second term in each pair of large parentheses as the possibility of a reported individual being a criminal subject to policing actions.
\end{proof}

\section{Proof of Proposition \ref{prop:scenario.1}}\label{appendix:mostly_innocent}
\begin{proof}
If $\xi_{d, u}\to 1$ for $d \in \{\majority, \minority\}$, then we have 
\begin{align*}
	\Delta(u) &=\Big (Y_{\innocent, \majo, u}\cdot \xi_{\majo, u}+Y_{\criminal, \majo, u}\cdot (1-\xi_{\majo, u})\Big )\\
 &-\Big ({Y_{\innocent, \mino, u}\cdot\xi_{\mino, u}} +Y_{\criminal, \mino, u}\cdot (1-\xi_{\mino, u})\Big )\\
	&\to Y_{\innocent, \majo, u}-Y_{\innocent, \mino, u}
\end{align*}
Thus, $\Delta(u)<0 \Leftrightarrow Y_{\innocent, \mino, u} > Y_{\innocent, \majo u}$.
\end{proof}

\section{Proof of Proposition \ref{prop:scenario.2}}\label{appendix:mostly_guilty}
\begin{proof}
If $\xi_{d, u}\to 0$ for $d \in \{\majority, \minority\}$, then we have 
\begin{align*}
	\Delta(u) &=\Big (Y_{\innocent, \majo, u}\cdot \xi_{\majo, u}+Y_{\criminal, \majo, u}\cdot (1-\xi_{\majo, u})\Big )\\
 &-\Big ({Y_{\innocent, \mino, u}\cdot\xi_{\mino, u}} +Y_{\criminal, \mino, u}\cdot (1-\xi_{\mino, u})\Big )\\
	&\to Y_{\criminal, \majo, u}-Y_{\criminal, \mino, u}
\end{align*}
Thus, $\Delta(u)<0 \Leftrightarrow Y_{\criminal, \mino, u} > Y_{\criminal, \majo u}$.
\end{proof}

\section{Proof of Proposition \ref{prop:scenario.3.1}}\label{appendix:proof3.1}
\begin{proof}

If $\frac{\xi_{\majo, u}}{\xi_{\mino, u}}\to 1$ and $\frac{Y_{\criminal, \majo, u}}{Y_{\criminal, \mino, u}}\to 1$, then we have 
\begin{align*}
	\Delta(u) &= \Big (Y_{\innocent, \majo, u}\cdot \xi_{\majo, u}+Y_{\criminal, \majo, u}\cdot (1-\xi_{\majo, u})\Big )\\
 &-\Big ({Y_{\innocent, \mino, u}\cdot\xi_{\mino, u}} +Y_{\criminal, \mino, u}\cdot (1-\xi_{\mino, u})\Big )\\
	&= \Big (Y_{\innocent, \majo, u}\cdot \xi_{\majo, u} - {Y_{\innocent, \mino, u}\cdot\xi_{\mino, u}}\Big )\\
 &-\Big (Y_{\criminal, \majo, u}\cdot (1-\xi_{\majo, u}) -Y_{\criminal, \mino, u}\cdot (1-\xi_{\mino, u})\Big )\\
	&= \Big ((Y_{\innocent, \majo, u}\cdot \xi_{\majo, u}-Y_{\innocent, \mino, u}\cdot \xi_{\majo, u} )\\
 &\quad\quad +(Y_{\innocent, \mino, u}\cdot \xi_{\majo, u} - {Y_{\innocent, \mino, u}\cdot\xi_{\mino, u}})\Big )\\
	&\quad\quad-\Big ((Y_{\criminal, \majo, u}-Y_{\criminal, \mino, u}) - (Y_{\criminal, \majo, u}\cdot\xi_{\majo, u}-Y_{\criminal, \mino, u}\cdot \xi_{\mino, u})\Big )\\
	&=(\xi_{\majo, u}\cdot (Y_{\innocent, \majo, u}-Y_{\innocent, \mino, u})+(Y_{\innocent, \mino, u}\cdot (\xi_{\majo, u} - \xi_{\mino, u})\Big )\\
	&\quad\quad-\Big ((Y_{\criminal, \majo, u}-Y_{\criminal, \mino, u}) - (Y_{\criminal, \majo, u}\cdot\xi_{\majo, u}-Y_{\criminal, \mino, u}\cdot \xi_{\mino, u})\Big )\\
	&\to \xi_{\majo, u}\cdot (Y_{\innocent, \majo, u}-Y_{\innocent, \mino, u})
\end{align*}
Thus, $\Delta(u)<0 \Leftrightarrow Y_{\innocent, \mino, u} > Y_{\innocent, \majo u}.$
\end{proof}

\section{Proof of Proposition \ref{prop:scenario.3.2}}\label{appendix:proof3.2}
\begin{proof}
If $\frac{\xi_{\majo, u}}{\xi_{\mino, u}}\to 1$ and $\frac{Y_{\innocent, \majo, u}}{Y_{\innocent, \mino, u}}\to 1$, then we have 
{\footnotesize
\begin{align*}
	\Delta(u) &= \Big (Y_{\innocent, \majo, u}\cdot \xi_{\majo, u}+Y_{\criminal, \majo, u}\cdot (1-\xi_{\majo, u})\Big )-\Big ({Y_{\innocent, \mino, u}\cdot\xi_{\mino, u}} +Y_{\criminal, \mino, u}\cdot (1-\xi_{\mino, u})\Big )\\
	&= \Big (Y_{\innocent, \majo, u}\cdot \xi_{\majo, u} - {Y_{\innocent, \mino, u}\cdot\xi_{\mino, u}}\Big )-\Big (Y_{\criminal, \majo, u}\cdot (1-\xi_{\majo, u}) -Y_{\criminal, \mino, u}\cdot (1-\xi_{\mino, u})\Big )\\
	&= \Big (Y_{\innocent, \majo, u}\cdot \xi_{\majo, u} - {Y_{\innocent, \mino, u}\cdot\xi_{\mino, u}}\Big )\\
	&\quad\quad-\Big ((Y_{\criminal, \majo, u}-Y_{\criminal, \mino, u}) - (Y_{\criminal, \majo, u}\cdot\xi_{\majo, u}-Y_{\criminal, \mino, u}\cdot \xi_{\mino, u})\Big )\\
	&=\Big (Y_{\innocent, \majo, u}\cdot \xi_{\majo, u} - {Y_{\innocent, \mino, u}\cdot\xi_{\mino, u}}\Big )\\
	&\quad\quad-\Big ((Y_{\criminal, \majo, u}-Y_{\criminal, \mino, u})\\
 &\quad\quad- (Y_{\criminal, \majo, u}\cdot\xi_{\majo, u}-Y_{\criminal, \majo, u}\cdot\xi_{\mino, u}+Y_{\criminal, \majo, u}\cdot\xi_{\mino, u}-Y_{\criminal, \mino, u}\cdot \xi_{\mino, u})\Big )\\
	&=\Big (Y_{\innocent, \majo, u}\cdot \xi_{\majo, u} - {Y_{\innocent, \mino, u}\cdot\xi_{\mino, u}}\Big )\\
	&\quad\quad-\Big ((Y_{\criminal, \majo, u}-Y_{\criminal, \mino, u}) - (Y_{\criminal, \majo, u}\cdot(\xi_{\majo, u}-\xi_{\mino, u})+\xi_{\mino, u}\cdot(Y_{\criminal, \majo, u}-Y_{\criminal, \mino, u})\Big )\\
	&\to (1-\xi_{\majo, u})\cdot (Y_{\criminal, \majo, u}-Y_{\criminal, \mino, u})
\end{align*}}%
Thus, $\Delta(u)<0 \Leftrightarrow Y_{\criminal, \mino, u} > Y_{\criminal, \majo u}.$
\end{proof}

\section{Proof of Proposition \ref{prop:scenario.3.3}}\label{appendix:proof3.3}
\begin{proof}
If $\frac{Y_{\criminal, \majo, u}}{Y_{\criminal, \mino, u}}\to 1$ and $\frac{Y_{\innocent, \majo, u}}{Y_{\innocent, \mino, u}}\to 1$, then we have
{\footnotesize
\begin{align*}
	\Delta(u) &= \Big (Y_{\innocent, \majo, u}\cdot \xi_{\majo, u}+Y_{\criminal, \majo, u}\cdot (1-\xi_{\majo, u})\Big )-\Big ({Y_{\innocent, \mino, u}\cdot\xi_{\mino, u}} +Y_{\criminal, \mino, u}\cdot (1-\xi_{\mino, u})\Big )\\
	&= \Big (Y_{\innocent, \majo, u}\cdot \xi_{\majo, u} - {Y_{\innocent, \mino, u}\cdot\xi_{\mino, u}}\Big )-\Big (Y_{\criminal, \majo, u}\cdot (1-\xi_{\majo, u}) -Y_{\criminal, \mino, u}\cdot (1-\xi_{\mino, u})\Big )\\
	&= \Big (Y_{\innocent, \majo, u}\cdot \xi_{\majo, u} - Y_{\innocent, \mino, u}\cdot \xi_{\majo, u} + Y_{\innocent, \mino, u}\cdot \xi_{\majo, u} -{Y_{\innocent, \mino, u}\cdot\xi_{\mino, u}}\Big )\\
	&\quad\quad-\Big ((Y_{\criminal, \majo, u}-Y_{\criminal, \mino, u}) - (Y_{\criminal, \majo, u}\cdot\xi_{\majo, u}-Y_{\criminal, \mino, u}\cdot \xi_{\mino, u})\Big )\\
	&=\Big ((Y_{\innocent, \majo, u}-Y_{\innocent, \mino, u})\cdot \xi_{\majo, u} + {Y_{\innocent, \mino, u}\cdot(\xi_{\majo, u}-\xi_{\mino, u})}\Big )\\
	&\quad\quad-\Big ((Y_{\criminal, \majo, u}-Y_{\criminal, \mino, u}) \\
 &\quad\quad- (Y_{\criminal, \majo, u}\cdot\xi_{\majo, u}-Y_{\criminal, \mino, u}\cdot\xi_{\majo, u}+Y_{\criminal, \mino, u}\cdot\xi_{\majo, u}-Y_{\criminal, \mino, u}\cdot \xi_{\mino, u})\Big )\\
	&=\Big ((Y_{\innocent, \majo, u}-Y_{\innocent, \mino, u})\cdot \xi_{\majo, u} + {Y_{\innocent, \mino, u}\cdot(\xi_{\majo, u}-\xi_{\mino, u})}\Big )\\
	&\quad\quad-\Big ((Y_{\criminal, \majo, u}-Y_{\criminal, \mino, u}) - \xi_{\majo, u}\cdot (Y_{\criminal, \majo, u}-Y_{\criminal, \mino, u})+Y_{\criminal, \majo, u}\cdot(\xi_{\majo, u}-\xi_{\mino, u})\Big )\\
	&\to  (Y_{\innocent, \mino, u}-Y_{\innocent, \majo, u})\cdot(\xi_{\majo, u}-\xi_{\mino, u})
\end{align*}}%
Thus, $\Delta(u)<0 \Leftrightarrow \xi_{\majo, u}>\xi_{\mino, u}.$
\end{proof}



\section{More Empirical Results} \label{appendix:statistics}
\subsection{Evaluating $\Delta(u)$}

In Section \ref{sec:empirics}, we present the results of evaluating \eqref{eq:delta.u} to further infer the presence or absence of parity in reporting and policing in real-world police-civilian interactions. For each incident, we associate $Y, M, R, U$ and $D$. We describe how we perform this association here.

\medskip 
\noindent {\em 911-calls $R$.} Every incident that is reported through 911 call corresponds to $R=1$, while incidents not reported through 911 call
correspond to $R=0$. {\color{black} As mentioned in Section \ref{Sec:parity_parity}, using data with $R =1$ is solely for inferring if parity in public reporting or parity in policing is violated, not introducing post-treatment conditioning error.}

\medskip
\noindent{\em Associating $M$.} Each 911 call either results into a police-civilian interaction, i.e. stop ($M=1$), or not ($M=0$). 

\medskip
\noindent{\em Associating $Y$.} The use of force in the stop incident determines value of $Y \in \{0,1\}$. In particular, for any threshold of use of force in the analysis (e.g. ``Handcuff''), if use of force is equal or more than that (e.g. ``Handcuff'' or more), then we set $Y = 1$, 
else we set $Y = 0$.

\medskip
\noindent{\em Associating $U$.} As per the \systemic framework, the context associated with incident can impact various aspects of the incident. We hypothesize that the context remains unchanged for a given precinct but may change across precincts. Therefore, we evaluate
the test statistics \eqref{eq:delta.u} for each precinct separately. 

\medskip
\noindent{\em Associating $D$.} Each incident that results into stop has race of the subject recorded in the data. Therefore, for all incidents with $M = 1$, 
we know $D$ precisely. However, for incidents reported in 911 calls but not resulting into stop, i.e. $R = 1$ and $M = 0$, the race is not available in
the data. This requires {\em inferring} race for such records. 

To that end, we assume for any incident type, the racial composition among incidents with $R=1, M=0$ and that among $R=1, M=1$ remains the same within a given context.

\begin{figure}[h!]
\centering\includegraphics[width=1\columnwidth]{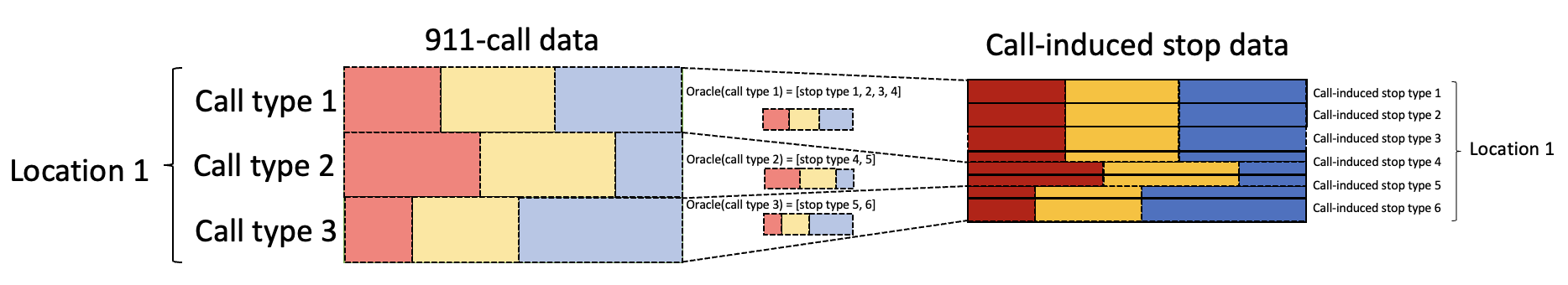}
\caption{Inference of race based on location and type 
}
\label{fig:race_inference_loc_type}
\end{figure}

Fig. \ref{fig:race_inference_loc_type} illustrates how the assumption works. For example, at context 1, as we know incidents with call type 1 often lead to stops with stop type 1, 2, 3 and 4, then we assume the racial composition of incidents with call type 1 matches with that of the incidents with stop type 1, 2, 3 and 4. {\color{black} Note that our assumption above is different from the claim that among the reported incidents $R=1$, race $D$ is independent of the stop decision $M$ given context $U$ as it is highly possible that calls leading to stops and calls not leading to stops can have vastly different compositions of incident types. } We acknowledge that even though such inference of race has utilized all the information embedded in the data available to us, it is clearly not yet perfect. For example, it is possible for some 911 call takers, if the suspects' races are mentioned in the calls, they may be more responsive to the calls involving certain races than other ones. We thus, as a next step, propose to collect recordings and scripts of 911 calls to recover the true labels of the races of reported individuals, thereby having more accurate estimations. 

We now provide an algorithm for this inference task. For each context, we repeat the following process:

{\color{black}
\medskip

\fbox{%
  \parbox{0.96\textwidth}{\indent For each call type:
\begin{enumerate}
	\item Find the matching stop types.
	\item Calculate the racial composition of all call-induced stops with the matching stop types from step (1) and use it as the racial composition of all calls with this call type.
	\item Estimate the number of involved civilians of different races for each call type based on the total number of calls with this call type and the racial composition from step (2).
\end{enumerate}
\indent Sum over all call types to calculate the overall racial composition at each context.
  }%
}
\\

With this algorithm, we are able to estimate the number of civilians of different races who are involved in the calls and stops. In this way, we are able to calculate the number of 911-reported majority and minority civilians and the number of punished majority and minority individuals, both of which are essential in our analysis.

In this way, the evaluation is based on empirical data and in effect
boils down to evaluating parameters of two Bernoulli random variables, $p(\majority, u) = \p[ Y=1, M=1 \,\mid\, R=1, D=\majority, U=u]$ and 
$p(\minority, u) = \p[ Y=1, M=1 \,\mid\, R=1, D=\majority, U=u]$. Let $\hat{p}(\majority, u)$ and $\hat{p}(\minority, u)$ be estimates. Then 
the standard errors are
\begin{align*}
\se(\majority, u) & = \sqrt{\hat{p}(\majority, u) (1-\hat{p}(\majority, u))} \nonumber \\
\se(\minority, u) & = \sqrt{\hat{p}(\minority, u) (1-\hat{p}(\minority, u))}.
\end{align*}
Therefore, we have 
\begin{align*}
\hat{\Delta}(u) & = \hat{p}(\majority, u) - \hat{p}(\minority, u)
\end{align*}
and its $95\%$ confidence interval is given by
\begin{align*}
	\hat{\Delta}(u) \pm 1.96 \sqrt{\frac{\se^2(\majority, u)}{\text{n}(\text{majority}, u)} + \frac{\se^2(\minority, u)}{\text{n}(\text{minority},u)}}
\end{align*}
 where $\text{n}(\text{majority}, u)$ and $\text{n}(\text{minority}, u)$ represent the number of 911-reported majority and minority civilians in a given context $U=u$ respectively.

The empirical results of NYC and New Orleans are presented in Section \ref{sec:empirics}. We also include the statistics summary of data below in Table \ref{table:stats_nyc} and \ref{table:stats_no}. We also present additional empirical results for Seattle.

 \begin{table}[!h]
\centering
\begin{tabular}{|c|c|c|c|}
\hline
Year                                            & 2018     & 2019     & 2020     \\ \hline
911 calls                       & 6,447,122  & 6,640,911  & 6,421,740  \\ \hline
Stops                           & 11,006    & 13,459    & 9,544     \\ \hline
Call-induced stops                    & 6,373     & 8,333     & 6,424     \\ \hline
$\mathbb{P}(\text{stop} \mid \text{call})$             & 0.10\%  & 0.13\%  & 0.10\%  \\ \hline
$\mathbb{P}(\text{call} \mid \text{stop})$    & 57.90\% & 61.91\% & 67.31\% \\ \hline
\end{tabular}
\caption{Basic statistics of NYC 911 call data and stop data.}
\label{table:stats_nyc}
\end{table}

\begin{table} [h!]
\centering
\scalebox{0.85}{
\begin{tabular}{|c|c|c|c|c|c|c|c|}
\hline
Year &  2014 & 2015 & 2016 & 2017 & 2018     & 2019     & 2020     \\ \hline
911 calls  &  303,599 & 307,563 & 314,266 & 325,389 & 310,304 & 327,690 & 271,868 \\ \hline
stops & 52,144 & 66,754 & 47,050 & 55,783 & 59,995 & 50,538 & 17,678\\ \hline
call-induced stops & 17,614 & 23,829 & 21,611 & 25,084 & 21,408 & 18,464 & 6,736\\ \hline
$\mathbb{P}(\text{stop} \mid \text{call})$ & 5.80\% &  7.75\% & 6.88\% & 7.71\% & 6.90\% & 5.63\% & 2.48\%\\ \hline
$\mathbb{P}(\text{call} \mid \text{stop})$ & 33.78\% & 35.70\% & 45.93\% & 44.97\% & 35.68\% & 36.53\% & 38.10\%\\ \hline
\end{tabular}}
\caption{Basic statistics of New Orleans 911 call and stop data.}
\label{table:stats_no}
\end{table}

 \subsection{Seattle} \label{appendix:seattle}

{\em Data.} Seattle is divided into 5 police districts, namely, North, East, South, West and Southwest. 
After data pre-processing, there are $722,863$
total 911 calls-for-service and $29,743$
stops over the years 2015 to 2019. Among the stops, $16,414$
are induced by 911 calls. 
In this sense, $2.27\%$ of 911 calls result in stops and $55.19\%$ are induced from by 911 calls. 
This again suggests that only a small portion of 911 calls end up in stops but the majority of stops are from 911 calls. 
The detailed year-to-year data statistics are 
summarized in Table~\ref{table:stats_seattle}. As before, 
we consider {\em white} as majority race while {\em Black} or {\em Hispanic} as minority race in the data analysis. We focus on two policing outcomes specified in the data, namely, frisk and arrest. 

\begin{table}[h]
\centering
\begin{tabular}{|c|c|c|c|c|c|}
\hline
Year &  2015 & 2016 & 2017 & 2018     & 2019     \\ \hline 
911 calls & 147,259 & 143,729 & 146,735 & 144,946 & 140,194 \\ \hline 
Stops  &  5,022 & 5,495 & 5,485 &  6,864 & 6,877 \\ \hline 
Call-induced stops & 2,797 & 2,981 & 3,104 & 3,744 & 3,788 \\ \hline 
$\mathbb{P}(\text{stop} \mid \text{call})$  & 1.90\%& 2.07\% & 2.12\% & 2.58\%& 2.70\% \\ \hline 
$\mathbb{P}(\text{call} \mid \text{stop})$ & 55.69\%& 54.25\% & 56.59\% & 54.55\% & 55.08\% \\ \hline 
\end{tabular}
\caption{Basic statistics of Seattle 911 call data and stop data.}
\label{table:stats_seattle}
\end{table}

\noindent{\em Findings.} The summary findings for Seattle are included in Fig. \ref{fig:seattle_pct_results}. The blue bars and green bars represent the proportions of precincts with statistically significant negative $\Delta(\cdot)$s and positive $\Delta(\cdot)$s. We remove the 2015 arrest, 2017, and 2019 White versus Hispanic results due to a lack of statistically significant results. 
For Seattle, like NYC, we find it fairly consistent that $\Delta(\cdot) < 0$, exhibiting observational bias in law enforcement against the minority. This seems to suggest that the primary source of such observational racial disparity against minorities could be bias in policing actions against minorities. 

\begin{figure}
\centering
\begin{subfigure}{\textwidth}
  \centering
  \includegraphics[width=0.8\linewidth]{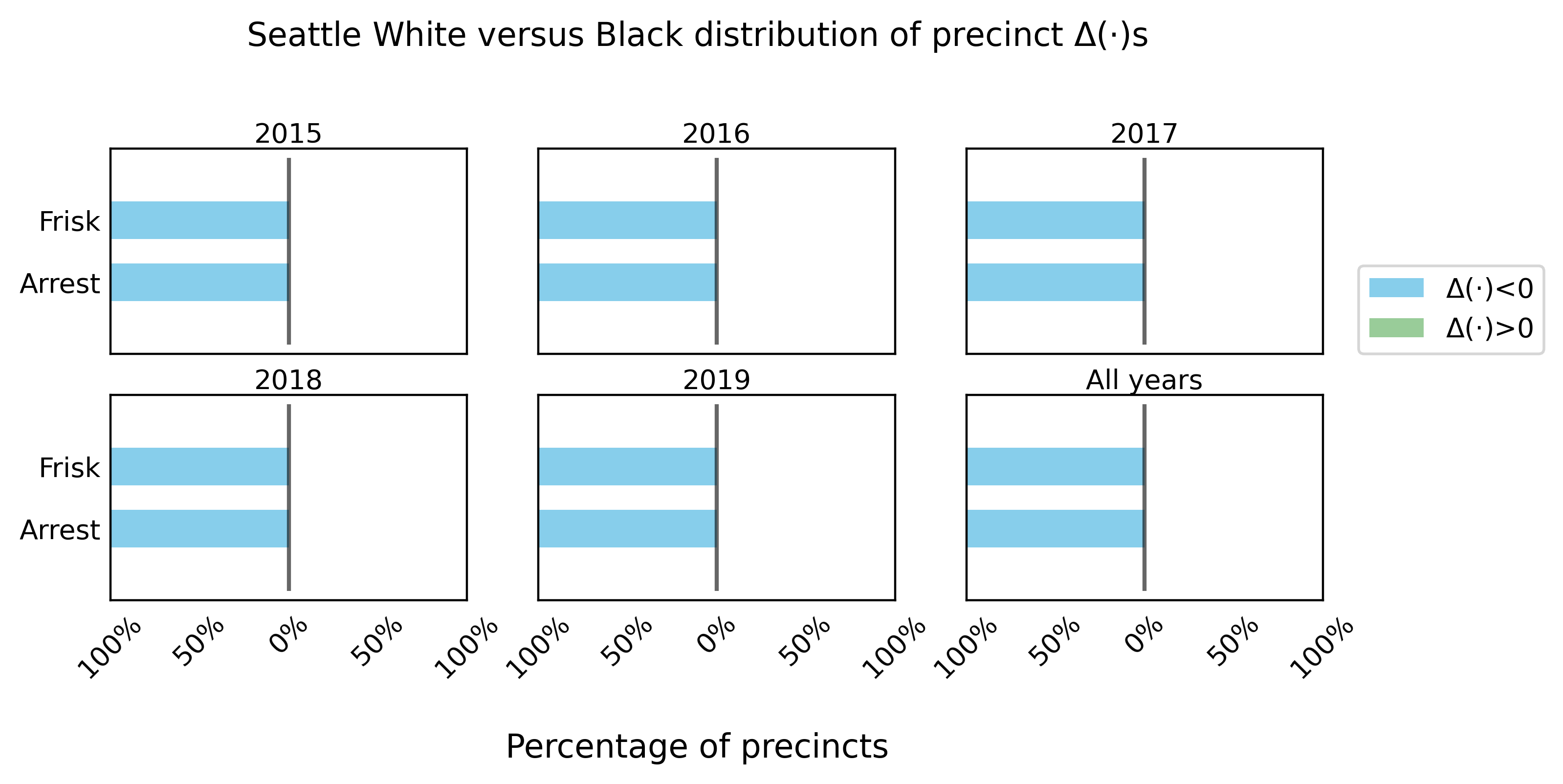}
  \caption{~}
  \label{fig:seattle_pct_black}
\end{subfigure}%

\begin{subfigure}{\textwidth}
  \centering
  \includegraphics[width=0.8\linewidth]{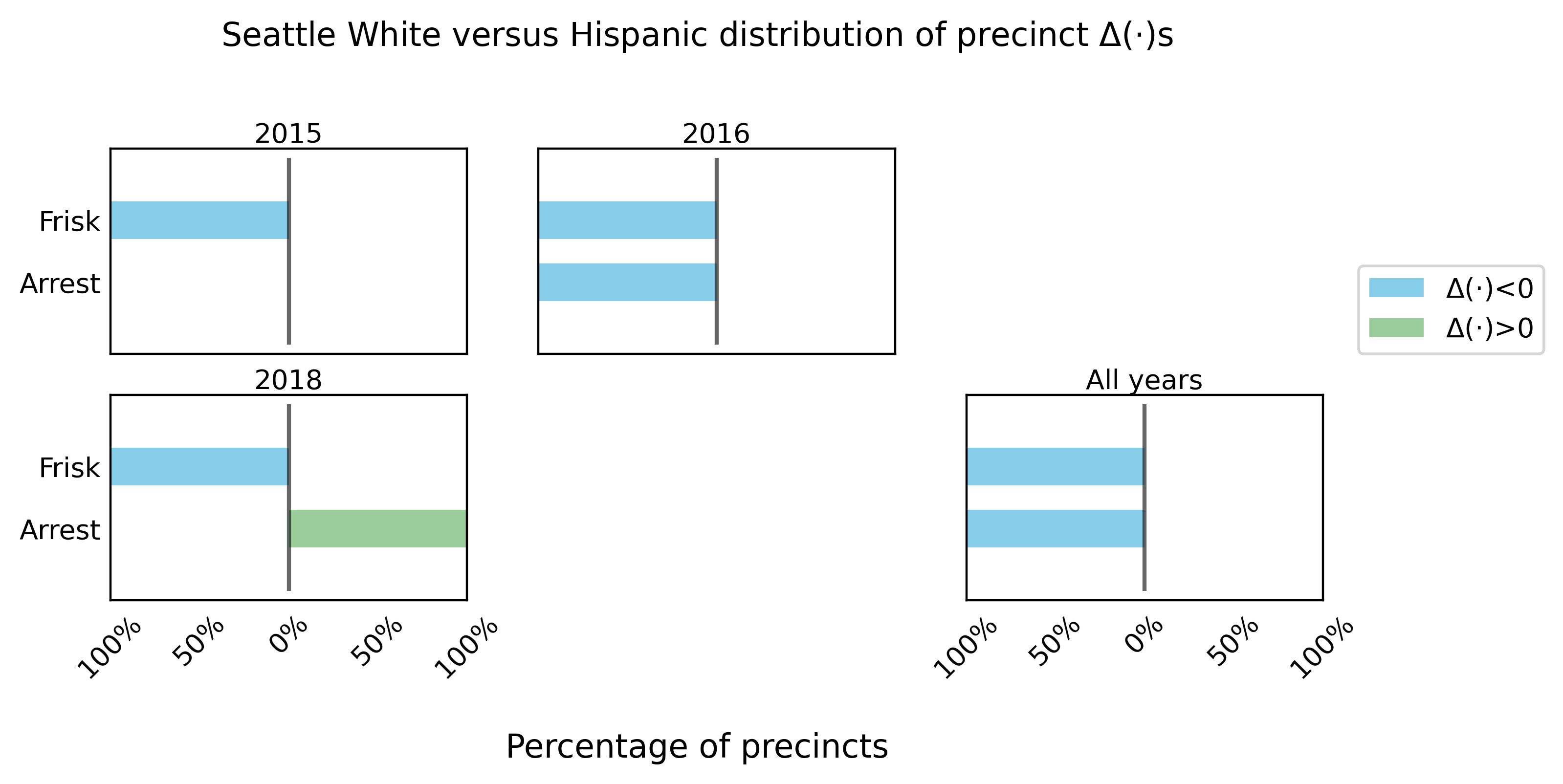}
  \caption{~}
  \label{fig:seattle_pct_hisp}
\end{subfigure}
\caption{A summary of statistically significant sign of $\Delta(\cdot)$ across
  all police districts in Seattle for different policing actions with (a) White versus Black, and (b) White versus Hispanic. Plots of 2015 Arrest, 2017 and 2019 in (b) are removed due to a lack of statistically significant results.}
\label{fig:seattle_pct_results}
\end{figure}

\end{document}